\documentclass[a4paper,fleqn]{cas-sc}

\usepackage[authoryear,longnamesfirst]{natbib}
\usepackage[ruled,vlined,linesnumbered,noresetcount]{algorithm2e}

\usepackage{hyperref}

\usepackage{array}
\newcolumntype{M}[1]{>{\centering\arraybackslash}m{#1}}
\newcolumntype{L}[1]{>{\raggedright\arraybackslash}m{#1}}
\newcolumntype{R}[1]{>{\raggedleft\arraybackslash}m{#1}}
\newcolumntype{N}{@{}m{0pt}@{}}

\usepackage{xcolor}
\usepackage{subcaption}

 \newtheorem{theorem}{Theorem}
  \newtheorem{definition}{Definition}
 \newproof{proof}{Proof}

\begin{document}
\let\WriteBookmarks\relax
\def\floatpagepagefraction{1}
\def\textpagefraction{.001}
\shorttitle{Proactive Defense against Advanced Persistent Threats}
\shortauthors{L. Huang and Q. Zhu}

\title [mode = title]{
A Dynamic Games Approach to Proactive Defense Strategies against Advanced Persistent Threats in Cyber-Physical Systems
}                      
%
%

\author[1]{Linan Huang}[
                        orcid=00000003-15918749
                        ]
\ead{lh2328@nyu.edu}
%
%
%
%
\author[1]{Quanyan Zhu}[%
   orcid=00000002-00082953
   ]
\ead{qz494@nyu.edu} 
%
%
\address[1]{Department of Electrical and Computer Engineering, New York University,
Brooklyn, NY, 11201}
%
%
%
%

\begin{abstract}
Advanced Persistent Threats (APTs) have recently emerged as a significant security challenge for a cyber-physical system due to their stealthy, dynamic and adaptive nature. Proactive dynamic defenses provide a strategic and holistic security mechanism to increase the costs of attacks and mitigate the risks. This work proposes a dynamic game framework to model a long-term interaction between a stealthy attacker and a proactive defender. 
The stealthy and deceptive behaviors are captured by the multi-stage game of incomplete information, where each player has his own private information unknown to the other. Both players act strategically according to their beliefs which are formed by the multi-stage observation and learning. The perfect Bayesian Nash equilibrium provides a useful prediction of both players' policies because no players benefit from unilateral deviations from the equilibrium. We propose an iterative algorithm to compute the perfect Bayesian Nash equilibrium and use the Tennessee Eastman process as a benchmark case study. Our numerical experiment corroborates the analytical results and provides further insights into the design of proactive defense-in-depth strategies. 
\end{abstract}

%


\maketitle

\section{Introduction}
The recent advances in automation technologies, 5G networks, and cloud services have accelerated the development of cyber-physical systems (CPSs) by integrating computing and communication functionalities with components in the physical world.
Cyber integration increases the operational efficiency of the physical system, yet it also creates additional security vulnerabilities. 
First, the increased connectivity and openness have expanded the attack surface \textcolor{black}{and enabled attackers to leverage vulnerabilities from multiple system components to launch a sequence of stealthy attacks.
}
Second, the component heterogeneity, the functionality complexity, and the \textcolor{black}{dimensionality} of cyber-physical systems have created many zero-day vulnerabilities, which make the defense \textcolor{black}{arduous and costly}. 

Advanced Persistent Threats (APTs) are \textcolor{black}{a class} of emerging threats for cyber-physical systems \textcolor{black}{with the following distinct features}. 
Unlike \textcolor{black}{opportunistic attackers who spray and pray}, APTs have specific targets and sufficient knowledge of the system architecture, valuable assets, and even defense strategies.  Attackers can tailor their strategies and invalidate cryptography, firewalls, and intrusion detection systems.  
Unlike \textcolor{black}{myopic attackers who smash and grab, APTs are stealthy and can disguise themselves as legitimate users for a long sojourn in the victim's system.}




\textcolor{black}{A few security researchers and experts have proposed APT models in which the entire intrusion process is divided into a sequence of phases, such as Lockheed-Martin's Cyber Kill Chain (see \cite{hutchins2011intelligence}), MITRE's ATT\&CK (see \cite{ATTACK}), the NSA/CSS technical cyber threat framework (see  \cite{DHS2018}), and the ones surveyed in \cite{messaoud2016advanced}.} 
Fig. \ref{fig: APT} illustrates the multi-stage structure of APTs.  During the reconnaissance phase, a threat actor \textcolor{black}{collects} open-source or internal  intelligence to identify valuable targets. 
After the attacker obtains a private key and establishes a foothold, he escalates privilege, propagates laterally in the cyber network, and eventually either accesses confidential information or inflicts physical damage. 
\textcolor{black}{Static standalone defense on a physical system cannot deter attacks originated from a cyber network.} 

\textcolor{black}{The multi-phase feature of APTs results in the concept of Defense in Depth (DiD), i.e.,
multi-stage cross-layer defense policies.  A system defender should adopt defensive countermeasures across the phases of APTs and holistically consider interconnections and interdependencies among these layers. 
To formally describe the interaction between an APT attacker and a defender with the defense-in-depth strategy, we map the sequential phases of APTs into a game of multiple stages.} Each stage describes a local interaction between the attacker and the defender where the outcome leads to the next stage of interactions. 
The goal of the attacker is to \textcolor{black}{stealthily} reach the targeted physical or informational assets while the defender aims to take defensive actions at multiple phases to thwart the attack or reduce its impact. 

\begin{figure*}
\centering
\includegraphics[width=0.95 \textwidth]{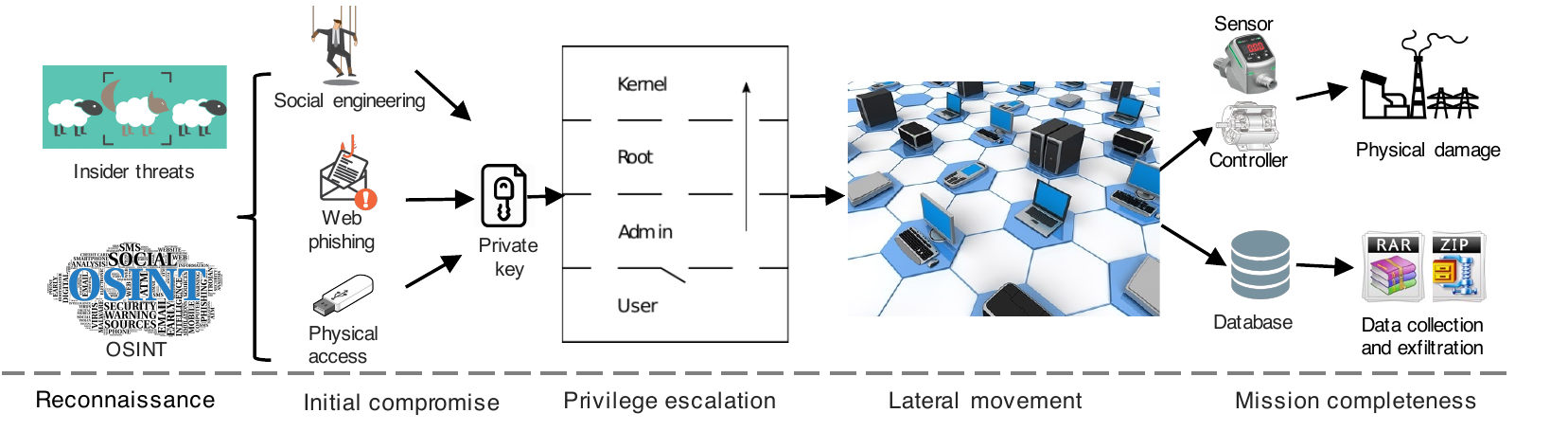}
\caption{
\textcolor{black}{An illustrate example} of the multi-stage structure of APTs. The multi-stage attack is composed of reconnaissance, initial compromise, privilege escalation, lateral movement, and mission execution.  \textcolor{black}{An attack originated from an early-stage cyber network can lead to damage in a physical system. }
 \label{fig: APT}}
\end{figure*}

\textcolor{black}{Detecting APTs timely (i.e., before attackers have reached the final stage) and effectively (i.e., with a low rate of false alarms and missed detections) is still an open problem  due to their stealthy and deceptive characteristics}. As reported in \cite{Jeff2000}, US companies in $2018$ have taken an average of $197$ and $69$ days, respectively, to detect and contain a data breach. \textcolor{black}{Stuxnet-like APT attacks can conceal themselves in a critical industrial system  for years and inconspicuously increase the failure probability of physical components. 
Due to the insufficiency of timely and effective detection systems for APTs, the defender remains uncertain about the user's type, i.e., either legitimate or adversarial, throughout stages. 
To prepare for the potential APT attacks, the defender needs to adopt precautions and proactive defense measures, which may also impair the user experience and reduce the utility of a legitimate user. Therefore, the defender needs to strategically balance the tradeoff between security and usability when the user's type remains private.   
}

\textcolor{black}{In this work, we model the private information of the user's type as a random variable following the work of \cite{harsanyi1967games}. 
Under the same defense action, the behavior and the utility of a user depend on whether his type is legitimate or adversarial. 
To make secure and usable decisions under incomplete information, the defender forms a belief on the user's type and updates the belief via the Bayesian rule based on the information acquired at each stage. 
For example,  throughout the phases of an APT, detection systems can generate many alerts based on suspicious user activities. 
Although these alerts do not directly reveal the user's type, a defender can use them to reduce the uncertainty on the user's type and better determine her defense-in-depth strategies at multiple stages.}

\textcolor{black}{Defensive deception provides an alternative perspective to bring uncertainty to the attacker and tilt the information asymmetry. 
We classify a defender into different levels of sophistication based on factors such as her level of security awareness, detection techniques she have adopted, and the completeness of her virus signature database. 
A sophisticated defender has a higher success rate of detecting adversarial behaviors. Thus, the behavior of an attacker depends on the type of defender that he interacts with. For example, the attacker may remain stealthy when he interacts with a sophisticated defender but behaves more aggressively when interacting with a primitive defender. 
As the attacker has incomplete information regarding the defender's type, he needs to form a belief and continuously updates it based on his observation of the defender's actions. In this way, the attacker can optimally decide whether, when, and to what extent, to behave aggressively or conservatively.  
}

\textcolor{black}{To this end, we also use a random variable to characterize the private information of the defender's type. As both players have incomplete information regarding the other player's type and they make sequential decisions across multiple stages, we extend the classical static Bayesian game to a multi-stage nonzero-sum game with two-sided incomplete information. }
Both players act strategically according to their beliefs to maximize their utilities. The Perfect Bayesian Nash Equilibrium (PBNE) provides a useful prediction of their policies at every stage for each type since no players can benefit from unilateral deviations at the equilibrium. Computing the PBNE is challenging due to the coupling between the forward belief update and the backward policy computation. 
We first formulate a mathematical programming problem to compute the equilibrium policy pair under a given belief for the one-stage Bayesian game. \textcolor{black}{For multi-stage Bayesian games, we compute the equilibrium policy pair under a given sequence of beliefs by constructing a sequence of nested mathematical programming problems.} 
Finally, we combine these programs with the Bayesian update and propose an efficient algorithm to compute the PBNE. 

The proposed modeling and computational methods are shown to be capable of hardening the security of a broad class of supervisory control and data acquisition (SCADA) systems. 
This work leverages the Tennessee Eastman  process as a case study of proactive \textcolor{black}{defense-in-depth strategies against the APT attackers who} can infiltrate into the cyber network through \textcolor{black}{phishing emails}, escalate privileges \textcolor{black}{through the process injection}, tamper the sensor reading \textcolor{black}{through malicious encrypted communication}, and \textcolor{black}{eventually} decrease the operational efficiency of the Tennessee Eastman process without triggering the alarm. 
The dynamic game approach offers a quantitative way to assess the risks and provides a systematic and computational mechanism to develop proactive and strategic defenses across multiple cyber and physical stages. 
\textcolor{black}{Based on the computation result of the case study, we obtain the following insights to guide the design of practical defense systems.} 
\begin{itemize}
    \item Defense at the final stage is usually too late to be effective when APTs have been well-prepared and ready to attack. 
We need to take precautions and proactive responses in the cyber stages when the attack remains ``under the radar" so that the attacker becomes less dominant when they reach the final stage. 
    \item The online learning capability of the defender plays an important role in detecting the adversarial deception and tilting the information asymmetry. It increases the probability of identifying the hidden information from the observable behaviors, threatens the stealthy attacker to take more conservative actions, and hence reduces the attack loss. 
    \item Third, defensive deception techniques are shown to be effective to \textcolor{black}{introduce uncertainty to attackers, increase their learning costs}, and hence reduce the probability of successful attacks. 
    \textcolor{black}{Those techniques may introduce a negative impact on legitimate users. However, a delicate balance between security and usability can be achieved under proper designs.}   
\end{itemize}

\subsection{Related Work}
    \textcolor{black}{One well-known industrial solution to APT defense is the ATT\&CK  matrix (see \cite{ATTACK}). It  illustrates disclosed attack methods and possible detection and mitigation countermeasures at different phases of APTs. 
However, as argued in \cite{endgame}, it lacks a prioritization to list all possible attack methods in one matrix. A lot of false alarms can arise as legitimate users can also generate a majority of activities in the ATT\&CK matrix. Besides, despite a persistent update, the matrix is far from complete and can lead to miss detection.} 

\textcolor{black}{
Many papers have attempted to deal with the above two challenges, i.e., false alarms and miss detection. 
To prevent security specialists from overwhelming alarms,  \cite{marchetti2016analysis} has analyzed high volumes of network traffic to reveal weak signals of suspect APT activities and ranked these signals based on the computation of suspiciousness scores. 
To identify attacks that exploit zero-day vulnerabilities or other unknown attack techniques, \cite{friedberg2015combating} has managed to learn and maintain a white-list of normal system behaviors and report all actions that are not on the white-list. 
There is also a rich literature on detecting essential components of an APT attack such as malicious PDF files in phishing emails (see \cite{nissim2015detection}), malicious SSL certificate during command and control communications (see \cite{ghafir2017malicious}), and data leakage at the final stage of the APT campaign (see \cite{sigholm2013towards}). 
These works have focused on a static detection of abnormal behaviors in one specific stage but had not taken into account the correlation among multiple phases of APTs. 
\cite{ghafir2018detection} has managed to build a framework to correlate alerts across multiple phases of APTs based on machine learning techniques so that all those alerts can be attributed to a single APT scenario. 
\cite{ghafir2019hidden} has constructed a correlation framework to link elementary alerts to the same APT campaign and applied the hidden Markov model to determine the most likely sequence of APT stages.  
}

\textcolor{black}{An alternative perspective from the aforementioned APT detection frameworks is to address how to respond to and mitigate potential attacks.
\cite{li2018defending} has captured the dynamic state evolution through a network-based  epidemic model and provided both prevention and recovery strategies for defenders based on optimal control approaches. 
Since APTs are controlled by human experts and can act strategically, the defender's response should adapt to the potential change of APT behaviors. Thus, decision and game theory becomes a natural quantitative framework to capture constraints on defense actions, attack consequences, and attackers' incentives.
 } 
 \cite{van2013flipit} has proposed \textit{FlipIt} game to model the key leakage under APTs as a private takeover between the system operator and the attacker. 
Many works have integrated \textit{FlipIt} with other components for the APT defense such as the signaling game to defend cloud service (see \cite{pawlick2018istrict}), an additional player to model the insider threats (see \cite{feng2015stealthy}), and a system of multiple nodes under limited resources (see \cite{zhang2015game}). 
The \textit{FlipIt} has described a high-level abstraction of the attacker's behavior to understand optimal timing for resource allocations. However, for our purpose of developing multi-stage defense policies, we need to provide a finer-grained model that can capture the dynamic interactions between players of different types across multiple stages. 
Our game framework models heterogeneous adversarial and defensive behaviors at multiple stages, allowing the prediction of attack moves and the estimation of losses using the equilibrium analysis.

\textcolor{black}{Other security game models such as \cite{zhu2018multi, yang2018effective, huang2017large} have provided dynamic risk management frameworks that allow the defender to response and repair effectively. 
In particular,  to model the multi-stage structure of APTs, \cite{zhu2018multi} has developed a sequence of heterogeneous game phases, i.e., a static Bayesian game for spear phishing, a nested game for penetration, and a finite zero-sum game for the final stage of physical-layer infrastructure protection. 
However, most of these security game frameworks have assumed complete information.
Our framework explicitly models the incomplete information across the entire phases of APTs and introduces their belief updates based on multi-stage information for making long-term strategic decisions. 
}



Cyber deception is an emerging research area. 
Games of incomplete information are natural frameworks to model the uncertainty and misinformation introduced by cyber deceptions. 
Previous works mainly focus on adversarial deceptions where the deceiver is the attacker. 
For example, strategic attackers in \cite{NguyenDeceptionIF} manipulate the attack data to mislead the defender in finitely repeated security games. 
A defender, on the other hand, can also initiate defensive deception techniques such as perturbations via external noises, obfuscations via revealing useless information, or honeypot deployments as shown in \cite{pawlick2017game}. 
\textcolor{black}{\cite{horak2017manipulating} proposes a framework to engage with attackers strategically to deceive them against the attack goal without their awareness.} 
A honeypot which appears to contain valuable information can lure attackers into isolation and surveillance. 
\textcolor{black}{\cite{la2016deceptive} has used a Bayesian game to model deceptive attacks and defenses in a honeypot-enabled network in the envisioned Internet of Things. 
Besides detection, a honeypot can also be used to obtain high-level indicators of compromise under a proper engagement policy as shown in \cite{engagement} where several security metrics are investigated and the optimal engagement policy is learned by reinforcement learning.} 
A system can also disguise a real asset as a honeypot to evade attacks as shown in \cite{rowe2007defending}. 
Our work considers a dynamic Bayesian game with double-sided incomplete information to incorporate both adversarial and defensive deceptions.  

The preliminary versions of this work (see \cite{huang2018analysis, huang2019adaptive}) have considered a dynamic game with one-sided incomplete information where attackers disguise themselves as legitimate users. 
This work extends the framework to a two-sided incomplete information structure where primitive systems can also disguise themselves as sophisticated systems. 
The new framework enables us to jointly investigate deceptions adopted by both attackers and defenders, and strategically design defensive deceptions to counter adversarial ones. 
We also develop new methodologies to address the challenge of the coupled belief update  
in a generalize setting without the previous assumption of the beta-binomial conjugate pair. 
In the case study, we investigate heterogeneous actions and cyber stages such as web phishing and privilege escalation, whose utilities are no longer negligible. 
Moreover, we leverage the Tennessee Eastman process with new performance metric and attack models to validate the efficacy of the proposed proactive defense-in-depth strategies, the Bayesian learning, and the defensive deception. 

\subsection{Organization of the Paper}
\textcolor{black}{We summarize notations, variables, and acronyms in Table. \ref{table:notation} for readers' convenience.} 
\begin{table}[]
\centering
\caption{Summary of notations, variables, and acronyms. 
\label{table:notation}}
\begin{tabular}{ll}
\hline
\textbf{General Notation}      & \textbf{Meaning}      \\ \hline
$A:=B$    & $A$ is defined as $B$     \\
$\Pr$      & Probability  \\
$f: A \mapsto B$      & A function or a mapping $f$  from domain $A$ to codomain $B$      \\ 
$\mathbb{E}_{a\sim A}[f(a)]$      & Expectation of $f(a)$ over random variable $a$ whose distribution is $A$   \\
$\mathbb{R}$      & Set of real numbers      \\
$|A|$     & The cardinality of set $A$   \\
$a\sim A$     & Random variable $a$ follows probability distribution $A$   \\
$\mathbf{1}_{\{x=y\}}$     & Indicator function which equals one when $x=y$, and zero otherwise   \\
$\{a_1,\cdots,a_n\}$ & Set with $n$ elements  $a_1,\cdots,a_n$\\
$B \setminus A$ & Set of elements in $B$ but not in $A$ 
\\
     &       \\
     
     \hline
\textbf{Variable} & \textbf{Meaning} \\ \hline
$i,j\in \{1,2\}$      & Index for players in the game: $i,j=1$ for the defender and $i,j=2$  for the user      \\
$\Theta_i$      & Set of all possible types of player   $i\in \{1,2\}$    \\
$\Delta (\Theta_i)$     & Space of probability distributions over type set $\Theta_i$ of player $i\in \{1,2\}$  \\
$\theta_i\in \Theta_i$      & Type of player  $i\in \{1,2\}$     \\
$\theta_1^H$ (resp. $\theta_1^L $)        & The defender is sophisticated (resp. primitive)     \\
$\theta_2^b$ (resp. $\theta_2^g$)       & The user is adversarial (resp. legitimate)     \\
$K$      & Total number of stages     \\
$k\in \{0,1,\cdots, K\}$      & Stage index       \\
$k_0\in \{0,1,\cdots, K\}$      & Index for the initial stage      \\
$A_i^k$      & Set of all possible actions of player $i\in \{1,2\}$ at stage  $k\in \{0,1,\cdots, K\}$       \\
$\Delta (A_i^k)$     & Space of probability distributions over the action set $A_i^k$ 
\\
$a_i^k\in A_i^k$      & Action of player $i\in \{1,2\}$ at stage  $k\in \{0,1,\cdots, K\}$       \\
$h^k, H^k$     & Action history and the set of all possible action histories at stage $k\in \{0,1,\cdots, K\}$      \\
$x^k, X^k$     & State and the set of all possible states at stage $k\in \{0,1,\cdots, K\}$      \\
$f^k$   & State transition function at stage $k$, i.e., $x^{k+1}=f^k(x^k,a_1^k,a_2^k)$ \\
$l_i^k, L_i^k$ & Available Information and set of all available information for player $i$ at stage  $k$ \\
$\sigma_i^k, \Sigma_i^k$      & Behavioral strategy and the set of all behavioral strategies   for player $i$ at stage  $k$     \\
$\sigma_i^k(a_i^k | l_i^k)$      &  Probability of player $i$ taking action $a_i^k$ at stage $k$ based on the available information $l_i^k$  \\
$\sigma_i^{k_0:K}$
 & Player $i$'s behavioral strategies from stage $k_0$ to $K$\\
 $\sigma_i^{*,k_0:K}$ ($\sigma_i^{*,K}:=\sigma_i^{*,K:K}$)  & Player $i$'s behavioral  strategies from stage $k_0$ to $K$ at the equilibrium \\
$b_i^k: L_i^k \mapsto \Delta (\Theta_j)$      & Player $i$'s belief on the other player $j$'s type at stage $k$ based on the available information  \\
$b_i^k(\theta_j | l_i^k)$      & Probability of player $j$ being type $\theta_j$ when player $i$ observes information $l_i^k$ at stage $k$  \\
\multirow{2}{*}{$\bar{J}_i^k(x^k,a_1^k,a_2^k,\theta_1,\theta_2, w_i^k)$} & Player $i$'s stage utility received at stage $k$ when the state is $x^k$, player $i$ takes action $a_i^k$,\\ 
&  \quad\quad \quad player $i$'s type is $\theta_i$, and the noise is $w_i^k$  \\
${J}_i^k(x^k,a_1^k,a_2^k,\theta_1,\theta_2 )$      & Player $i$'s expected stage utility received at stage $k$ with the input of  $x^k,a_1^k,a_2^k,\theta_1,\theta_2$\\
\multirow{2}{*}{${U}_i^{k_0:K}(\sigma_i^{k_0:K},\sigma_j^{k_0:K}, x^{k_0},\theta_i )$}      & Player $i$'s expected cumulative utility received from stage $k_0$ to $K$ when the initial state \\
& \quad\quad\quad is $x^{k_0}$, his/her type is $\theta_i $, and the multi-stage strategies of player $i$ are  $\sigma_i^{k_0:K}$\\
$V_i^k(x^k,\theta_i)$ & Player $i$'s value function at state $x^k$ when his/her type is $\theta_i$\\
     &       \\
\hline
\textbf{Acronym}      & \textbf{Meaning}      \\ \hline
APT(s)      & Advanced persistent threat(s)      \\
SBNE      & Static Bayesian Nash equilibrium       \\
DBNE      & Dynamic  Bayesian Nash equilibrium    \\
PBNE      & Perfect Bayesian Nash equilibrium     \\ \hline
\end{tabular}
\end{table}
We use pronoun `he' for the user and `she' for the defender throughout this paper. 
The rest of the paper is organized as follows. Section \ref{sec:model} introduces the multi-stage game with incomplete information and three equilibrium concepts are defined in Section \ref{sec:analysis}.
To compute these equilibria, we construct constrained optimization problems and an iterative algorithm in Section \ref{sec:computation}. 
A case study of Tennessee Eastman process under APTs  
is presented in Section \ref{sec: Casestudy} with results in Section \ref{sec:result}. Section \ref{sec:conclusion} concludes the paper.
\section{Dynamic Game Modelling of APT Attacks}
\label{sec:model}
\textcolor{black}{There are two players in the game, player $1$ is the user and player $2$ is the defender.} 
The stealthy, persistent, and deceptive features of APTs result in   incomplete information \textcolor{black}{of the user's type} to the defender. 
\textcolor{black}{We use a finite set $\Theta_2$ to accommodate all possible types of the user.} 
For example, we consider a binary type set for the case study in Section \ref{sec: Casestudy} and \ref{sec:result} where the user's type $\theta_2$ is either adversarial $\theta_2^b$ or legitimate $\theta_2^g$. 
\textcolor{black}{The APT attacker, i.e., the adversarial user,} disguises himself as the legitimate user, thus the defender does not know the type of the user. 
\textcolor{black}{The set of the user's type can also be non-binary and incorporate different APT groups when their attack tools and targeted assets are different (see \cite{FireEye2017})}. 

\textcolor{black}{The Defender can also be classified into different levels of sophistication based on various factors such as her level of security awareness, detection techniques she adopted, and the completeness of her virus signature database. 
The discrete type $\theta_1$ distinguishes defenders of different sophistication levels and all the possible type values constitute the defender's type set $\Theta_1$.} 
For example, in our case study, the defender's type $\theta_1$ is either  sophisticated $\theta_1^H$ or primitive $\theta_1^L$.
\textcolor{black}{The defender can apply defensive deception techniques and keep her type private to the user. 
We assume that both players' type sets are commonly known. Each player knows his/her own type, yet not the other player's type. Thus, each player $i$ should treat the other player's type as a random variable with an initial distribution $b_i^0$ and update the distribution to $b_i^k$ when obtaining new information at each stage $k$. We present the above belief update formally in Section \ref{sec:beliefupdate}.} 
\subsection{Multi-stage Transition}
\textcolor{black}{We formulate the interaction between the multi-stage APT attack and the cross-stage proactive defense into $K$ stages of sequential games with incomplete information, as shown in Fig. \ref{fig: logicflow}.}
\begin{figure}
\centering
\includegraphics[width=1 \textwidth]{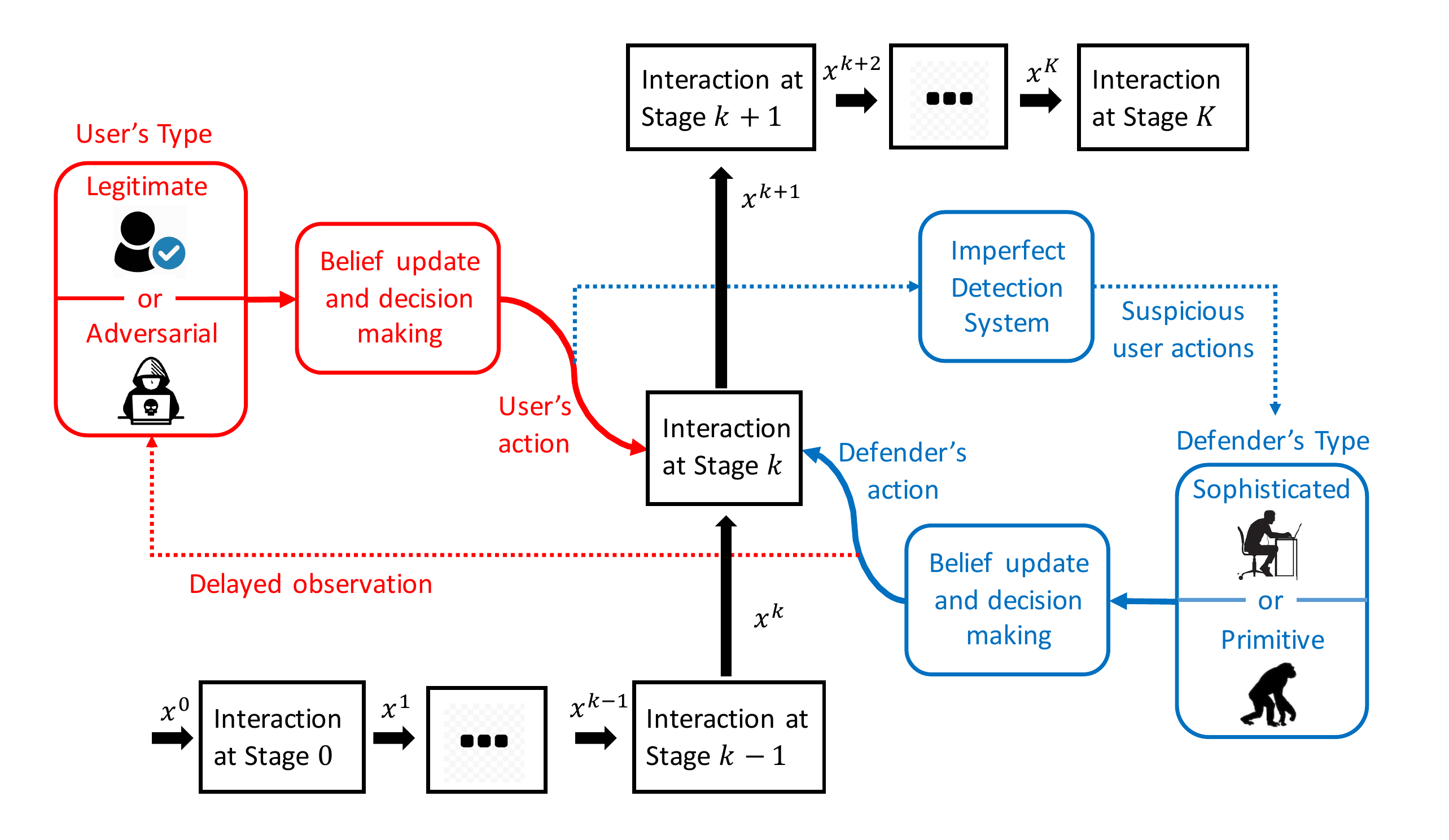}
\caption{
\textcolor{black}{A block diagram of applying the defense-in-depth approach against multi-stage APT attacks. We denote the user, the defender, and the system states in red, blue, and black, respectively. 
The defender interacts with the user from stage $0$ to stage $K$ in sequence where the output state of stage $k-1$ becomes the input state of stage $k$. 
At each stage $k$, the user observes the defender's actions at previous stages, forms a belief on the defender's type, and takes an action. At the same time, the defender makes decisions based on the output of an imperfect detection system. 
The dotted line means that the observation is not in real time, i.e., both players can only observe the previous-stage actions of the other player. }
 \label{fig: logicflow}}
\end{figure}
At each stage $k\in \{0,1,\cdots,K\}$, player $i\in \{1,2\}$ takes an action $a_i^k\in {A}_i^k$ from a finite and discrete set ${A}_i^k$. 
\textcolor{black}{An intrusion detection system generates alerts based on the user's actions. 
However, since legitimate users can also trigger these alerts, each alert  itself does not reveal the user's type. 
For example, an APT attacker uses the Tor network connection for data exfiltration, yet a legitimate user can also use it legally for the traffic confidentiality as shown in \cite{milajerdi2015composite}. 
Another example is that code obfuscation can be either used legitimately to prevent reverse engineering or illegally to conceal malicious JavaScript code from being recognized by signature-based detectors or human analysts as shown in \cite{nissim2015detection}.  
We assume that the user can observe the defender's stage-$k$ action at stage $k+1$.  
The observation of the defender's action at a single stage also does not reveal the defender's type.} 

\textcolor{black}{In this paper, each player obtains a one-stage delayed observation of the other player's actions, i.e., at each stage $k$, the action history available to both players is $h^k=\{a_1^0,\cdots,a_1^{k-1},a_2^0,\cdots,a_2^{k-1}\} \in {H}^k:=\prod_{i=1}^2 \prod_{\bar{k}=0}^{k-1} {A}_i^{\bar{k}}$.} 
 Given history $h^k$ at the current stage $k$, players at stage $k+1$ obtain an updated history
$h^{k+1}=h^k\cup \{a_1^k,a_2^k\}$ after the observation of both players' actions at stage $k$. 
At each stage $k$, we further define a state $x^k\in {X}^k$  which summarizes information about both players' actions in previous stages so that the initial state $x^0\in {X}^0$ and the history at stage $k$ uniquely determine ${x}^k$ through a known state transition function $f^k$, i.e., $x^{k+1}=f^k(x^k,a_1^k,a_2^k), \forall k\in \{0,1,\cdots,K-1\}$. 
\textcolor{black}{
States at different stages can have different meanings} such as \textcolor{black}{the reconnaissance outcome}, the user's location,  the privilege level, and the sensor status.    


\subsection{Behavioral Strategy}
\textcolor{black}{
A defender should behave differently when interacting with  adversarial users and legitimate ones. The defensive measure should also vary for attackers who adopt different code families and tools. 
However, since the defender is uncertain about the user's type throughout the entire stages of games, she has to make judicious decisions at each stage to balance usability versus security. 
The user's action should also adapt to the type of the defender. For example, if the defender is primitive, an attacker prefers to take aggressive adversarial actions to achieve a quicker and low-cost compromise. However, if the defender is sophisticated and can detect the malware with  better accuracy, an attacker  has to take conservative actions to remain stealthy. 
Since the proactive defense actions across the entire stages can affect legitimate users, they also need to be designed to avoid collateral damage.} 

\textcolor{black}{Thus, the decision-making problem of the defender or the user boils down to the determination of} a behavioral strategy  $\sigma^k_i \in \Sigma^k_i:{L}^k_i \mapsto \Delta ({A}_i^k)$, i.e.,  player $i$ at each stage $k$ needs to \textcolor{black}{decide which action to take or take an action with what probability based on the information $l_i^k\in L_i^k$ available to him/her at stage $k$.}  We present two different information structures in Section \ref{sec:info1} and  \ref{sec:info3}. 
\textcolor{black}{
The strategy is called `behavioral' as the strategy depends on the information available at the time the players make their decisions. 
In this work, players are allowed to take \textit{mixed strategies}, thus the co-domain of the strategy function $\sigma^k_i$ is $\Delta ({A}_i^k)$, a probability distribution over the action space ${A}_i^k$.} 
 With a slight abuse of notation, we denote $\sigma^k_i(a_i^k| {l}_i^k)$ as the probability of player $i$ taking action $a_i^k\in {A}_i^k$ given the available information $l_i^k\in {L}_i^k$. The actual action of player $i$ taken at stage $k$, i.e., $a_i^k$, is a realization of the behavioral strategy $\sigma_i^k$. 
Note that the values of the other player's type $\theta_j$ and action $a_j^k$ are not observable for player $i$ at stage $k$, thus do not affect player $i$'s behavioral strategy $\sigma_i^k$, i.e., $\Pr({a}_i^k|{a}^k_{j},\theta_{j}, {l}_i^k)=\sigma_i^k({a}_i^k| {l}_i^k)$.  Therefore, $\sigma_1^k$ and $\sigma_2^k$ are conditionally independent, i.e., $\Pr({a}_i^k,{a}^k_{j}|{l}_i^k,{l}_j^k)=\sigma^k_i({a}_i^k| {l}_i^k)\sigma^k_{j}({a}^k_{j}| {l}_j^k)$. 

\subsection{Belief and Bayesian Update}
\label{sec:beliefupdate}
\textcolor{black}{To quantify the uncertainty of the other player's type throughout the entire stages}, each player $i$ forms a belief $b^k_i: {L}_i^k  \mapsto \Delta (\Theta_{j}), j\neq i$. 
Likewise, \textcolor{black}{$b^k_i(\theta_j| {l}_i^k)$ means that} given information $l_i^k\in {L}_i^k$ at stage $k$, player $i$ forms a belief that the other player $j$ is of type $\theta_j\in \Theta_j$ with probability $b^k_i(\theta_j| {l}_i^k)$. 
At the initial stage $k=0$, the only information \textcolor{black}{available to player $i$ is his/her own type, i.e., $l_i^0=\theta_i$.} 
We assume that player $i$ has a prior belief distribution $b_i^0$ based on the past experiences with the other player. 
If no previous experiences are available to player $i$, player $i$ can take the uniform distribution as an unbiased prior belief.  
\textcolor{black}{As each player $i$ obtains new information when arriving at the next stage, his or her belief can be updated using the Bayesian rule. 
 We present the Bayesian update under} two different information structures ${L}_i^k$ at stage $0<k\leq K$ in the following two subsections. 

\subsubsection{Timely Observations}
\label{sec:info1}
The most straightforward information structure is ${L}_i^k={H}^k\times \Theta_i$\textcolor{black}{, i.e., the information available to player $i$ at stage $k$ is the action history $h^k$ and player $i$'s own type $\theta_i$,} which leads to the belief update in \eqref{eq: history-dependent belief update}. 
\begin{equation}
\begin{split}
&{b}^{k+1}_i(\theta_{j}|h^{k}\cup \{a_i^k,a_j^k\}, \theta_i)=\allowbreak
\frac{   \sigma^k_i({a}_i^k| {h}^k, \theta_i)\sigma^k_{j}({a}^k_{j}| {h}^k,\theta_{j}) b_i^k(\theta_{j}| {h}^k, \theta_i)    }
{\sum_{\bar{\theta}_{j}\in \Theta_j} \sigma^k_i({a}_i^k| {h}^k, \theta_i)\sigma^k_{j}({a}^k_{j}| {h}^k,\bar{\theta}_{j}) b_i^k(\bar{\theta}_{j}| {h}^k, \theta_i)   }, i,j\in \{1,2\},j\neq i. 
\end{split}
\label{eq: history-dependent belief update}
\end{equation}
Here, player $i$ updates the belief ${b}^{k}_i$ based on the observation of the action $a_i^k,a_j^k$. When the denominator is $0$, the history $h^{k+1}$ is not reachable from $h^k$, and the Bayesian update does not apply. In this case, we let 
 ${b}^{k+1}_i(\theta_{j}|h^{k}\cup \{a_i^k,a_j^k\}, \theta_i):=b_i^0(\theta_j |\theta_i)$.

\subsubsection{Markov Belief}
\label{sec:info3}
\textcolor{black}{
If the information available to player $i$ at stage $k$ is the state value $x^k$ and player $i$'s own type $\theta_i$, then the information set is taken to be ${L}_i^k={X}^k\times \Theta_i$.} 
With the Markov property that $\Pr(x^{k+1}| \theta_{j}, x^k, \cdots,x^1,x^0, \theta_i)=\Pr(x^{k+1}| \theta_{j}, x^k, \theta_i)$, the Bayesian update between two consequent states is 
\begin{equation}
\label{eq: state-dependent belief update}
\begin{split}
&b_i^{k+1}(\theta_{j} | x^{k+1},\theta_i)=\allowbreak
\frac{\Pr(x^{k+1}| \theta_{j}, x^k,\theta_i)b_i^{k}(\theta_{j}|x^k,\theta_i)}{\sum_{\bar{\theta}_{j}\in \Theta_j} \Pr(x^{k+1}| \bar{\theta}_{j}, x^k,\theta_i)b_i^{k}(\bar{\theta}_{j}|x^k,\theta_i)}, i,j\in \{1,2\},j\neq i.
\end{split}
\end{equation}


With the conditional independence of $\sigma^k_1$ and $\sigma^k_2$, 
\begin{equation}
\label{eq: cluster}
\begin{split}
& \Pr(x^{k+1}| \theta_{j}, x^k,\theta_i)\allowbreak 
=\sum_{\{a^k_1,a^k_2\}\in \bar{A}^k}\sigma^k_1(a^k_1|x^k,\theta_1)\sigma^k_2(a^k_2|x^k,\theta_2), 
\end{split}
\end{equation}
 where  $\bar{A}^k:=\{a^k_1\in {A}^k_1,a^k_2\in {A}^k_2| x^{k+1}=f^k(x^k,a^k_1,a^k_2)\}$ contains all the action pairs that change the system state from $x^k$ to $x^{k+1}$.  
Equation \eqref{eq: cluster} shows that the Bayesian update in \eqref{eq: state-dependent belief update} can be obtained from \eqref{eq: history-dependent belief update} by clustering all the action pairs in set \textcolor{black}{$\bar{A}^k$}. 
Thus, the Markov belief update \eqref{eq: state-dependent belief update} can also be regarded as an approximation of \eqref{eq: history-dependent belief update} using action aggregations. 
Unlike the history set ${H}^k$, the dimension of the state set,  $|{X}^k|$,  does not grow with the number of stages. Hence, the Markov approximation significantly reduces the
 memory and computational complexity. 
The following sections adopt the Markov belief update. 


\subsection{Stage and Cumulative Utility}
\label{subsec: utility}
\textcolor{black}{
The player's utility can vary under the same action taken by different types of users or defenders. 
For example, the remote access from a legitimate teleworker brings a reward to the defender while the one from an adversarial user inflicts a loss.}
Therefore, at each stage $k$, player $i$'s stage utility $\bar{J}_i^k: {X}^k \times {A}_1^k \times {A}_2^k \times \Theta_1 \times \Theta_2 \times \mathbb{R} \mapsto \mathbb{R} $ can depend on both players' types and actions, the current state ${x}^k\in {X}^k$, and an external noise $w_i^k\in \mathbb{R}$ with a known probability density function $\varpi_i^k$. 
The noise term models unknown or uncontrolled factors that can affect the value of the stage utility. 
The existence of the external noise makes it impossible for player $i$, after reaching stage $k+1$,  to infer the value of the other player's type $\theta_j$ based on the knowledge of the input parameters $x^k , a^k_1 , a^k_2 , \theta_i$,  together with the output of the utility function $\bar{J}_i^k$ at stage $k$. 

We denote the expected stage utility as $J_i^k(x^k, a^k_1 , a^k_2 , \theta_1,\theta_2):=\mathbb{E}_{w_i^k\sim \varpi_i^k }[\bar{J}_i^k(x^k, a^k_1 , a^k_2 , \theta_1,\theta_2,w_i^k)], \forall x^k, a^k_1 , a^k_2 , \theta_1,\theta_2$. 
Given the  type $\theta_i\in \Theta_i$, the initial state $x^{k_0}\in {X}^{k_0}$, and both players' strategies $\sigma_i^{k_0:K}:=[\sigma^k_i(a_i^k|{x}^{k},\theta_i)]_{k=k_0,\cdots,K}\in \prod_{k=k_0}^K \Sigma_i^k$ from stage $k_0$ to $K$, we can determine the expected cumulative utility $U_i^{k_0:K}$ for player $i$, 
i.e., 
\begin{equation}
\label{eq: cumultive utility}
\begin{split}
&U^{k_0:K}_i(\sigma_i^{k_0:K},\sigma_{j}^{k_0:K}, x^{k_0},\theta_i) \allowbreak 
:=\sum_{k=k_0}^{K} \mathbb{E}_{\theta_{j}\sim b_i^k, a_i^k\sim \sigma_i^k,a_j^k \sim \sigma_{j}^k} [J_i^k({x}^k,a_1^k,a_{2}^k,\theta_1,\theta_{2})]\\
&=
\sum_{k=k_0}^K \sum_{\theta_{j}\in \Theta_{j}} {b}^{k}_i(\theta_{j}|{x}^k, \theta_i)     \sum_{a_{i}^k \in {A}_{i}^k} \sigma_{i}^k(a_{i}^k |{x}^k, \theta_{i}) \cdot \allowbreak \sum_{a_{j}^k \in {A}_{j}^k}\sigma_{j}^k(a_{j}^k |{x}^k, \theta_{j}) J_i^k({x}^k,a_1^k,a_{2}^k,\theta_1,\theta_{2}),\  j\neq i. 
\end{split}
\end{equation}


\section{PBNE and Dynamic Programming}
\label{sec:analysis}
The user and the defender use the Bayesian update to reduce their uncertainties on the other player's type.  Since their actions affect the belief update, both players at each stage should optimize their expected cumulative utilities concerning 
the updated beliefs \textcolor{black}{at the future stages}, which leads to the Perfect Bayesian Nash Equilibrium (PBNE) in Definition \ref{def: PBNE}. 


\begin{definition}
\label{def: PBNE}
Consider the two-person $K$-stage game with double-sided incomplete information (i.e., each player's type is not known to the other player), a sequence of beliefs $b_i^k, \forall k\in \{0,\cdots, K\}$, an expected cumulative utility $U^{0:K}_i$ in \eqref{eq: cumultive utility}, and a given scalar $\varepsilon\geq 0$. A sequence of strategies $\sigma_i^{*,0:K}\in \prod_{k=0}^K \Sigma_i^k$ is called 
$\varepsilon$-dynamic Bayesian Nash equilibrium for player $i$  if condition (C2) is satisfied. If condition (C1) is also satisfied, 
$\sigma_i^{*,0:K}$ is further called 
$\varepsilon$-perfect Bayesian Nash equilibrium. 
\begin{enumerate}
\item[{(C1)}] Belief consistency: under strategy pair $(\sigma_1^{*,0:K}, \sigma_2^{*,0:K})$, each player's belief $b_i^{k}$ at each stage $k=0,\cdots, K$ satisfies 
\eqref{eq: state-dependent belief update}. 
\item[{(C2)}] Sequential rationality: for all given initial state $x^{k_0}\in {X}^{k_0}$ at every initial stage $k_0\in \{0,\cdots,K\}$, 
\begin{equation}
\label{eq: PBNE in def}
\begin{split}
&U_1^{k_0:K}(\sigma_1^{*,k_0:K},\sigma_{2}^{*,k_0:K}, {x}^{k_0},\theta_1)+\varepsilon
\geq \allowbreak
U_1^{k:K}(\sigma_1^{k_0:K},\sigma_{2}^{*,k_0:K}, {x}^{k_0},\theta_1),  \forall \sigma_{1}^{k_0:K}\in \prod_{k=0}^K \Sigma_1^k; 
\\
&U_2^{k_0:K}(\sigma_1^{*,k_0:K},\sigma_{2}^{*,k_0:K}, {x}^{k_0},\theta_2)+\varepsilon
\geq   \allowbreak
U_2^{k:K}(\sigma_1^{*,k_0:K},\sigma_{2}^{k_0:K}, {x}^{k_0},\theta_2), \forall \sigma_{2}^{k_0:K}\in \prod_{k=0}^K \Sigma_2^k.  
\end{split}
\end{equation}
\end{enumerate}
When  $\varepsilon = 0$, the two $\varepsilon$-equilibria are called Dynamic Bayesian Nash Equilibrium ({DBNE}) and  Perfect Bayesian Nash Equilibrium ({PBNE}), respectively. 
\end{definition}

\textcolor{black}{The belief consistency emphasizes that when strategic players make long-term decisions, they have to consider the impact of their actions on their opponent's beliefs at future stages.} 
The PBNE is a refinement of the DBNE with the additional requirement of the belief consistency property. 
When the horizon $K=0$, the multi-stage game of incomplete information defined in Section \ref{sec:model} degenerates to a one-stage (static) Bayesian game with the one-stage belief pairs $(b_1^K, b_2^K)$ and the solution concept of the DBNE/PBNE degenerates to the Static Bayesian Nash Equilibrium ({SBNE}) in Definition \ref{Def: BNE}. 

The sequential rationality property in \eqref{eq: PBNE in def} guarantees that unilateral deviations from the equilibrium at any states 
do not benefit the deviating player. Thus, the equilibrium strategy can be a reasonable prediction of both players' multi-stage behaviors. 
DBNE strategies have the property of \allowbreak 
\textit{strongly time consistency} because \eqref{eq: PBNE in def} holds for any possible initial states, even for states that are not on the equilibrium path, i.e., those states would not be visited under DBNE strategies. 
The \textit{strongly time consistency} property makes the DBNE adapt to unexpected changes. 
Solutions obtained by dynamic programming naturally satisfy \textit{strongly time consistency}. Hence, in the following, we introduce algorithms based on dynamic programming techniques.

Define the value function $V_i^{k_0}({x}^{k_0},\theta_i)  
:= \allowbreak 
U_i^{k_0: K}(\sigma_1^{*,k_0:K},\sigma_{2}^{*, k_0:K}, x^{k_0},\theta_i)$ as the utility-to-go from any initial stage $k_0\in \{0,\cdots,K\}$ under the DBNE strategy pair $(\sigma_1^{*,k_0:K},\sigma_{2}^{*, k_0:K})$. 
Then, at the final stage $K$, the value function for player $i \in \{1,2\}$  with type $\theta_i$ at state $x^K$ is
\begin{equation}
\label{eq: finalstage}
\begin{split}
V_i^K( x^{K},\theta_i)= \sup_{\sigma_i^{K} \in \Sigma_i^{K} }   
\mathbb{E}_{\theta_{j}\sim b_i^K, a_i^K\sim \sigma_i^{K},a_j^K \sim \sigma_{j}^{*,K} } \allowbreak
 [J_i^K({x}^K,a_1^K,a_{2}^K,\theta_1,\theta_{2})]. 
\end{split}
\end{equation}
 
For any feasible sequence of belief pairs $(b_1^k,b_2^k), \allowbreak
{k=0,   
\cdots,
K-1}$, we have the following recursive system equations  for player $i$ to find the equilibrium strategy pairs $(\sigma_1^{*,k},\sigma_2^{*,k})$ backwardly from stage $K-1$ to the initial stage $0$, i.e., $\forall k\in \{0,\cdots,K-1\}, \forall i,j\in \{1,2\}, j\neq i$, 
\begin{equation}
\label{eq: backward DP}
\begin{split} 
& V_i^{k}({x}^{k},\theta_i)= \sup_{\sigma_i^{k} \in \Sigma_i^{k} }  
 \mathbb{E}_{\theta_{j} \sim b_i^{k}, a_i^{k}\sim \sigma_i^{k},a_j^{k} \sim \sigma_{j}^{*,{k}}} \allowbreak
 [ V_i^{k+1}(f^k({x}^{k},a_1^{k},a_{2}^{k}),\theta_i) 
+  J_i^{k}(x^{k}, a_1^{k},a_{2}^{k},\theta_1,\theta_{2})]. 
 \end{split}
\end{equation}
If we assume a virtual termination value \allowbreak 
${V}_i^{K+1}(f^K(x^{K},a_1^K,a_2^K),\theta_i) \equiv 0$, we can obtain \eqref{eq: finalstage} by letting stage $k=K$ in \eqref{eq: backward DP}. 
The second term in \eqref{eq: backward DP} represents the immediate stage utility and the first term represents the expected utility under the future state $x^{k+1}= f^k({x}^{k},a_1^{k},a_{2}^{k}), \allowbreak 
k\in \{0,\cdots,K-1\}$. 
Since $a_i^k$ affects both terms, players should adopt a long-term perspective and avoid myopic behaviors to  balance between the immediate utility and the expected future utility.

%
%
%

\section{Computational Algorithms}
\label{sec:computation}

In \ref{sec: Static Bay}, we formulate a constrained optimization problem to compute the SBNE and $V_i^K$ for the one-stage game. 
In  \ref{sec: dynamic Bay}, we use the proposed optimization problem as building blocks to compute the DBNE and  $V_i^k, \forall k\in \{0,\cdots, K-1\}$. Finally, we propose an iterative algorithm to solve for the PBNE. 
Efficient algorithms to compute the PBNE lay a solid foundation to quantify the risk of cyber-physical attacks and guide the design of proactive defense-in-depth strategies. 


\subsection{One-Stage Bayesian Game and SBNE}
\label{sec: Static Bay}
Since both players' actions at the final stage $k=K$  only affect the immediate utility $J_i^K$ and there is no future state transition, we can  treat the final-stage game at each state $x^K\in {X}^K$ as an equivalent one-stage Bayesian game with the belief $b_i^K$ and obtain the SBNE. 
\begin{definition}
\label{Def: BNE}
A pair of mixed-strategies $(\sigma_1^{*,K}\in \Sigma_1^{K}, \sigma_2^{*,K}\in \Sigma_2^{K})$ is said to constitute a Static Bayesian Nash Equilibrium (SBNE) under the given belief pair $(b_1^K  ,b_2^K) $ and the state $x^K\in {X}^K$, if $\forall \theta_1\in \Theta_1, \theta_2\in \Theta_2$, 
\textcolor{black}{
\begin{equation}
\label{eq: BNE matrix form}
\begin{split}
& \mathbb{E}_{\theta_{2}\sim b_1^K, a_1^K\sim \sigma_1^{*,K},a_2^K \sim \sigma_{2}^{*,K} } \allowbreak [ J_1^K({x}^K,a_1^K,a_{2}^K,\theta_1,\theta_{2})]
\geq
\mathbb{E}_{\theta_{2}\sim b_1^K, a_1^K\sim \sigma_1^{K},a_2^K \sim \sigma_{2}^{*,K} } \allowbreak [ J_1^K({x}^K,a_1^K,a_{2}^K,\theta_1,\theta_{2})], \forall \sigma_1^{K}\in \Sigma_1^{K}; \\
& \mathbb{E}_{\theta_{1}\sim b_2^K, a_1^K\sim \sigma_1^{*,K},a_2^K \sim \sigma_{2}^{*,K} } \allowbreak [ J_2^K({x}^K,a_1^K,a_{2}^K,\theta_1,\theta_{2})]
\geq
\mathbb{E}_{\theta_{1}\sim b_2^K, a_1^K\sim \sigma_1^{*,K},a_2^K \sim \sigma_{2}^{K} } \allowbreak [ J_2^K({x}^K,a_1^K,a_{2}^K,\theta_1,\theta_{2})], \forall \sigma_2^{K}\in \Sigma_2^{K}. \\
\end{split}
\end{equation} 
}
\end{definition}
In Theorem \ref{Thm: double-sided program}, we propose a constrained optimization program ${C}^K$ to compute the SBNE. We suppress the superscript of $K$ without any ambiguity in one-stage games.

\begin{theorem}
\label{Thm: double-sided program}
A strategy pair  $(\sigma_1^{*}\in \Sigma_1, \sigma_2^*\in \Sigma_2)$ constitutes a SBNE to the one-stage bi-matrix Bayesian game $({J}_1, {J}_2)$ under private type $\theta_i\in \Theta_i, \forall i\in \{1,2\}$,  belief $b_i, \forall i\in \{1,2\}$, and a given state $x$, if and only if  
the strategy pair is a solution to ${C}^K$:  
\textcolor{black}{
\begin{equation*}
\begin{split}
 {[{C}^K]}:  & \max_{\sigma_{1},\sigma_{2},s_1,s_2}  \  
 \sum_{\theta_1\in \Theta_1} \alpha_1(\theta_1)s_1(x,\theta_1)+  \sum_{\theta_2 \in \Theta_2} \alpha_2(\theta_2)s_2(x,\theta_2) \\
&\quad +  \sum_{\theta_1\in \Theta_1} 
 \alpha_1(\theta_1)\mathbb{E}_{\theta_{2}\sim b_1, a_1\sim \sigma_1,a_2 \sim \sigma_{2} } \allowbreak [ J_1(x,a_1,a_{2},\theta_1,\theta_{2})]
\\
&\quad + \sum_{\theta_2 \in \Theta_2 } 
 \alpha_2(\theta_2)  \mathbb{E}_{\theta_{1}\sim b_2, a_1\sim \sigma_1,a_2 \sim \sigma_{2} } \allowbreak [ J_2({x},a_1,a_{2},\theta_1,\theta_{2})]
 \\
\text {s.t.}  \quad \quad   &(a)  \quad  
 \mathbb{E}_{\theta_{1}\sim b_2, a_1\sim \sigma_1 } \allowbreak [ J_2({x},a_1,a_{2},\theta_1,\theta_{2})]  \leq -s_2(x,\theta_2) , \forall \theta_2, \forall a_2, \\
 & (b) \quad 
\sum_{a_1\in A_1} \sigma_1(a_1|x,\theta_1)=1, \sigma_1(a_1|x,\theta_1)\geq 0,  \forall \theta_1, \\
&(c)  \quad  
\mathbb{E}_{\theta_{2}\sim b_1,a_2 \sim \sigma_{2} } \allowbreak [ J_1({x},a_1,a_{2},\theta_1,\theta_{2})] \leq -s_1(x,\theta_1), \ \forall \theta_1, \forall a_1, \\
& (d) \quad 
\sum_{a_2\in A_2} \sigma_2(a_2|x,\theta_2)=1, \sigma_2(a_2|x,\theta_2)\geq 0,  \forall \theta_2.
\label{eq: constrained optimization}
\end{split}
\end{equation*}
}
\textcolor{black}{The dimensions of decision variables $\sigma_1(a_1|x,\theta_1), \forall \theta_1\in \Theta_1,$ and $\sigma_2(a_2|x,\theta_2), \forall \theta_2\in \Theta_2,$ are $|A_1|\times |\Theta_1|$ and $|A_2|\times |\Theta_2|$, respectively. Besides,  $s_1(x,\theta_1), \forall \theta_1$ and $s_2(x,\theta_2), \forall \theta_2$  are scalar decision variables for each given $\theta_i, i\in \{1,2\}$. }
The \textcolor{black}{non-decision variables}  $\alpha_1(\theta_1), \forall \theta_1$ and $\alpha_2(\theta_2), \forall \theta_2,$ can be any strictly positive and finite numbers. 
The solution to $C^K$ exists and is achieved at the equality of constraints $(a), (c)$, i.e., $s_2^*(x,\theta_2)= -V_2(x,\theta_2)
, s_1^*(x,\theta_1)= -V_1(x,\theta_1)$. 
\end{theorem}

\begin{proof}
\label{proof}
The finiteness and discreteness of the action and the type spaces guarantee the existence of the SBNE in mixed strategies as shown in  \cite{shoham2008multiagent}, which further guarantee that program ${C}^K$ has solutions. 
To show the equivalence between the solution to ${C}^K$ and the SBNE, we first show that every SBNE is a solution of ${C}^K$. 
If $(\sigma_1^{*}\in \Sigma_1, \sigma_2^*\in \Sigma_2)$ is a SBNE pair, then 
the quadruple  
$
\sigma_1^*(\theta_1), \sigma_2^{*}(\theta_2), 
s_2^*(x,\theta_2)=-V_2(x,\theta_2), \allowbreak
s_1^*(x,\theta_1)=-V_1(x,\theta_1), \forall \theta_i \in \Theta_i, 
\forall i\in \{1,2\}
$, 
is feasible because it satisfies constraints $(a),(b),(c),\allowbreak
(d)$. 
Constraints $(a)$ and $(c)$ imply a non-positive objective function of ${C}^K$. 
Since the value of the objective function achieved under this quadruple is  $0$, this quadruple is also optimal. 
Second, we show that $\sigma_1^*(\theta_1), \sigma_2^{*}(\theta_2), 
s_2^*(x,\theta_2), 
s_1^*(x,\theta_1)$, the result of ${C}^K$ is a SBNE.  
The solution of  ${C}^K$ should satisfy all the constraints, i.e., 
\textcolor{black}{
\begin{equation}
\label{eq: feasible}
\begin{split}
&  \mathbb{E}_{\theta_{1}\sim b_2, a_1\sim \sigma_{1}^*,a_2\sim \sigma_{2}} \allowbreak [ J_2({x},a_1,a_{2},\theta_1,\theta_{2})]  \leq -s^*_2(x,\theta_2) , \forall \theta_2,\forall \sigma_2\in \Sigma_2,\\
&  \mathbb{E}_{\theta_{2}\sim b_1, a_1\sim \sigma_{1},a_2\sim \sigma_{2}^*} \allowbreak [ J_2({x},a_1,a_{2},\theta_1,\theta_{2})]  \leq -s^*_1(x,\theta_1) , \forall \theta_1, \forall \sigma_1\in \Sigma_1.
\end{split}
\end{equation}
}
In particular, if we pick $\sigma_i(\theta_i)=\sigma_i^*(\theta_i), \forall \theta_i, \forall i\in \{1,2\}$, \textcolor{black}{and combine the  fact} that the optimal value is achieved at $0$, the inequality turns out to be an equality and equation \eqref{eq: feasible} becomes \eqref{eq: BNE matrix form}, which shows that $(\sigma_1^{*}\in \Sigma_1, \sigma_2^*\in \Sigma_2)$  is a SBNE. 
\qed
\end{proof}


Theorem \ref{Thm: double-sided program} focuses on the double-sided Bayesian game where each player player $i$ has a private type $\theta_i\in \Theta_i$. 
To accommodate the one-sided Bayesian game where player $i$'s type ${\theta}_i\in\Theta_i$ is known by both players and player $j$'s type remains unknown to player $i$,  we can modify program ${C}^K$  
by letting $\alpha_i({\theta}_i)>0$ and \textcolor{black}{$\alpha_i(\tilde{\theta}_i)=0, \forall \tilde{\theta}_i \in \Theta_i \setminus \{ {\theta}_i \}$.} 

\subsection{Multi-stage Bayesian Game and PBNE}
\label{sec: dynamic Bay}
From \eqref{eq: backward DP}, we can see that at stages $k<K$, each player optimizes the sum of the immediate utility $J_i^k$ and the utility-to-go $V_i^k$. 
Thus, we can replace the original stage utility $J_i^K$ in program ${C}^K$ with $V_i^k+J_i^k$ in program ${C}^k$ to compute the DBNE in a multi-stage Bayesian game. 
\begin{theorem}\label{Thm: PBNE}
Given a sequence of beliefs $b_i^k$ for each player $i\in \{1,2\}$ at each stage $k\in \{0,1,\cdots,K-1\}$, a strategy pair $({\sigma}_1^{*,0:K-1},{\sigma}_2^{*,0:K-1})$ constitutes a  DBNE of the $K$-stage Bayesian game under double-sided incomplete information with the expected cumulative utility $U^{0:K}_i$ in \eqref{eq: cumultive utility}, if, and only if  ${\sigma}_1^{*,k},{\sigma}_2^{*,k}$, $s_1^{*,k}(x^k,\theta_1),s_2^{*,k}(x^k, \theta_2)$ are the optimal solutions to the following constrained optimization problem ${C}^k$ for each $k\in \{0,1,\cdots,K-1\}$:
\begin{equation*}
\begin{split}
{[{C}^k]}:  & \max_{{\sigma}^k_1,{\sigma}_2^k,s^k_1,s^k_2} 
\quad 
\sum_{i=1}^2  
\sum_{\theta_i\in \Theta_i} \alpha_i(\theta_i)
\{s^k_i(x^k,\theta_i)+ \sum_{\theta_j\in \Theta_j}  
b_i^k (\theta_j | x^k, \theta_i) \sum_{a_1^k\in {A}_1^k}{\sigma}^k_1(a_1^k|x^k, \theta_1)\sum_{a_2^k\in {A}_2^k}{\sigma}^k_2(a_2^k|x^k, \theta_2)  \\
& \quad \cdot [{J}_i^k(x^k,a_1^k,a_2^k, \theta_{1},\theta_{2}) +{V}_i^{k+1}(f^k(x^{k},a_1^k,a_2^k),\theta_i)]
\}
\\
 \text{ s.t. }  &  (a)   \quad
\sum_{\theta_1\in \Theta_1} b_2^k (\theta_1| x^k, \theta_2)\sum_{a_1^k\in {A}_1^k} {\sigma}^k_1(a_1^k|x^k,\theta_1) \cdot [{J}_2^k(x^k,a_1^k,a_2^k, \theta_{1},\theta_{2}) +{V}_2^{k+1}(f^k(x^{k},a_1^k,a_2^k),\theta_2)] \\
& \quad \quad \leq -s^k_2(x^k,\theta_2), \forall \theta_2\in \Theta_2, \forall a_2^k\in A_2^k, 
\\
& (b) \quad
\sum_{\theta_2\in \Theta_2} b_1^k (\theta_2|x^k, \theta_1)\sum_{a_2^k\in {A}_2^k}{\sigma}^k_2(a_2^k|x^k,\theta_2) \cdot  [{J}_1^k(x^k,a_1^k,a_2^k, \theta_{1},\theta_{2})+{V}_1^{k+1}(f^k(x^{k},a_1^k,a_2^k),\theta_1)]\\
& \quad \quad \leq -s^k_1(x^k,\theta_1) , \forall \theta_1\in \Theta_1, \forall a_1^k\in A_1^k.
\end{split}
\end{equation*}
Similarly,  $\alpha_1(\theta_1),\alpha_2(\theta_2)$ can be any strictly positive and finite numbers, and 
$(s_1^{k}(x^k,\theta_1),s_2^{k}(x^k,\theta_2))$ 
is a sequence of {scalar variables} for each $x^k\in X^k, \theta_i\in \Theta_i, i\in \{1,2\}$. 
The optimum exists and is achieved at the equality of constraints $(a),(b)$, i.e., $s_i^{*,k}(x^k, \theta_i)=-V_i^k(x^k,\theta_i), \forall \theta_i\in \Theta_i, \forall i\in \{1,2\}$. 
\end{theorem}

The proof is similar to the one for Theorem \ref{Thm: double-sided program}. 
The decision variables ${\sigma}^k_i$ are of size $|{A}_i^k |\times |{X}^k| \times |\Theta_i|$. 
By letting stage $k=K$ and $V_i^{K+1}=0$, program ${C}^K$ for the static Bayesian game is a special case of ${C}^k$ for the multi-stage Bayesian game. 
We can solve program ${C}^{k+1}$ to obtain the DBNE strategy pair $({\sigma}^{k+1}_1,{\sigma}_2^{k+1})$ and the value of ${V}_i^{k+1}$. Then, we apply ${V}_i^{k+1}$ in program ${C}^k$ to obtain a DBNE strategy pair $({\sigma}^{k}_1,{\sigma}_2^{k})$ and the value of ${V}_i^{k}$. 
Thus, for any given sequences of type belief pairs $b_i^k, \forall i\in \{1,2\}, \forall k\in \{0,1,\cdots,K\}$, we can solve ${C}^k$ from $k=K$ to $k=0$ recursively to obtain the DBNE pair $({\sigma}_1^{*,0:K-1},{\sigma}_2^{*,0:K-1})$. 

\subsubsection{PBNE}
\label{sec:PBNErefine}
Given a sequence of beliefs, we can obtain the corresponding DBNE  via ${C}^k$ in a backward fashion. However, given a sequence of policies, both players  forwardly update their beliefs at each stage by \eqref{eq: state-dependent belief update}. Thus, we need to find a consistent pair of belief and policy sequences \textcolor{black}{as required by the PBNE}. 
As summarized in Algorithm \ref{algorithm}, we  iteratively alternate between the forward belief update and the backward policy computation to find the PBNE.  
We resort to $\varepsilon$-PBNE solutions when the existence of PBNE is not guaranteed. 


\begin{algorithm}[h]
\SetAlgoLined
 Initialization beliefs $b_i^k$ at each stage $k\in \{0,1,\cdots,K\}$, \textsc{IterNum}$>0$, $\varepsilon\geq 0$. 
 
 \While{the $t<$\textsc{IterNum} }{
 $t:=t+1$\;
 \For{ each $x^K\in {X}^K$ }
 {
  Compute SBNE strategy ${\sigma}_i^{*,K}$ and $V_i^K(x^K,\theta_i)$ via ${C}^K$.
  }
  \For{$k\leftarrow K-1$ \KwTo $0$}
  { 
   \For{each $x^k\in {X}^k$ }
 {
  Compute DBNE strategy ${\sigma}_i^{*,k}$ and $V_i^k(x^k,\theta_i)$ via ${C}^k$. 
  }
  }
   \For{$k\leftarrow 0$ \KwTo $K-1$}
  { 
  Update $b_i^k$ with ${\sigma}_i^{*,0:K-1}$ via  \eqref{eq: state-dependent belief update}. 
  }
  \uIf{${\sigma}_i^{*,0:K-1}, \forall i \in \{1,2\}$,  satisfy \eqref{eq: PBNE in def}}{
\textbf{Terminate}\
  }
  }
  \textbf{Output}  $\varepsilon$-PBNE strategy pair $(\sigma_1^{*,0:K-1},\sigma_2^{*,0:K-1})$ and consistent beliefs $b_i^k, \forall k\in \{0,\cdots,K\}$. 
 \caption{Numerical Solution of $\varepsilon$-PBNE\label{algorithm}}
\end{algorithm}

Algorithm \ref{algorithm} provides a computational approach to find $\varepsilon$-PBNE with the following procedure. 
First, both players initialize their beliefs $b_i^k$ for every state $x^k$ at stage $k\in \{0,1,\cdots,K\}$, according to their types. 
Then, they compute the DBNE strategy pair $\sigma_i^{*,0:K}, \forall i\in \{1,2\}$, under the given belief sequence at each stage by solving program ${C}^k$ from stage $K$ to stage $0$ in sequence. 
Next, they update their beliefs at each stage according to the strategy pair $\sigma_i^{*,0:K-1}, \forall i\in \{1,2\}$, via the Bayesian update \eqref{eq: state-dependent belief update}. 
If the strategy pair $\sigma_i^{*,0:K-1}, \forall i\in \{1,2\}$, satisfies \eqref{eq: PBNE in def} under the updated belief, we find the $\varepsilon$-PBNE and terminate the iteration. Otherwise, we repeat the backward policy computation in step two and the forward belief update in step three.


\section{Case Study}
\label{sec: Casestudy}
\textcolor{black}{The model presented in Section \ref{sec:model} can be applied to various APT scenarios. To illustrate the framework, this section presents a specific attack scenario} where the attacker stealthily initiates infection and escalates privileges in the cyber network, aiming to launch attacks on the physical plant as shown in Fig. \ref{fig: StateDiagram}. 
\textcolor{black}{Three vertical columns in the left block illustrate the state transitions across three stages: the initial compromise, the privilege escalation, and the sensor compromise of a physical system. 
The red squares at each column represent possible states at that stage. 
The right block illustrates a simplified flow chart of the Tennessee Eastman Process.} 
We use the Tennessee Eastman process as a benchmark of industrial control systems to show that attackers can  strategically compromise the SCADA system and decrease the operational efficiency of a physical plant without triggering the alarm. 

In this case study, we adopt the binary type space $\Theta_2=\{\theta_2^b, \theta_2^g\}$ and $\Theta_1=\{\theta_1^H, \theta_1^L\}$ for the user and the defender, respectively. 
In particular, $\theta_2^b$ and $\theta_2^g$ denote the adversarial and legitimate user, respectively; $\theta_1^H$ and $\theta_1^L$ denote the sophisticated and primitive defender, respectively. 
The bi-matrices in Table \ref{table: multiBay-initial}, \ref{table: muti-stageBay}, and \ref{Table: FinalMatrix} represent both players' expected utilities at three stages, respectively. 
In these matrices, the defender is the \textit{row player} and the user is the \textit{column player}. 
Each entry of the matrix corresponds to players' payoffs under their action  pairs, types, and the state. 
In particular, the two red numbers in the parenthesis before the semicolon are the payoffs of the defender and the user, respectively, under type $\theta_2^b$, while the parenthesis  in blue after the semicolon presents the payoff of the defender and the user, respectively, under type $\theta_2^g$. 

\begin{figure*}
\centering
\includegraphics[width=0.95 \textwidth]{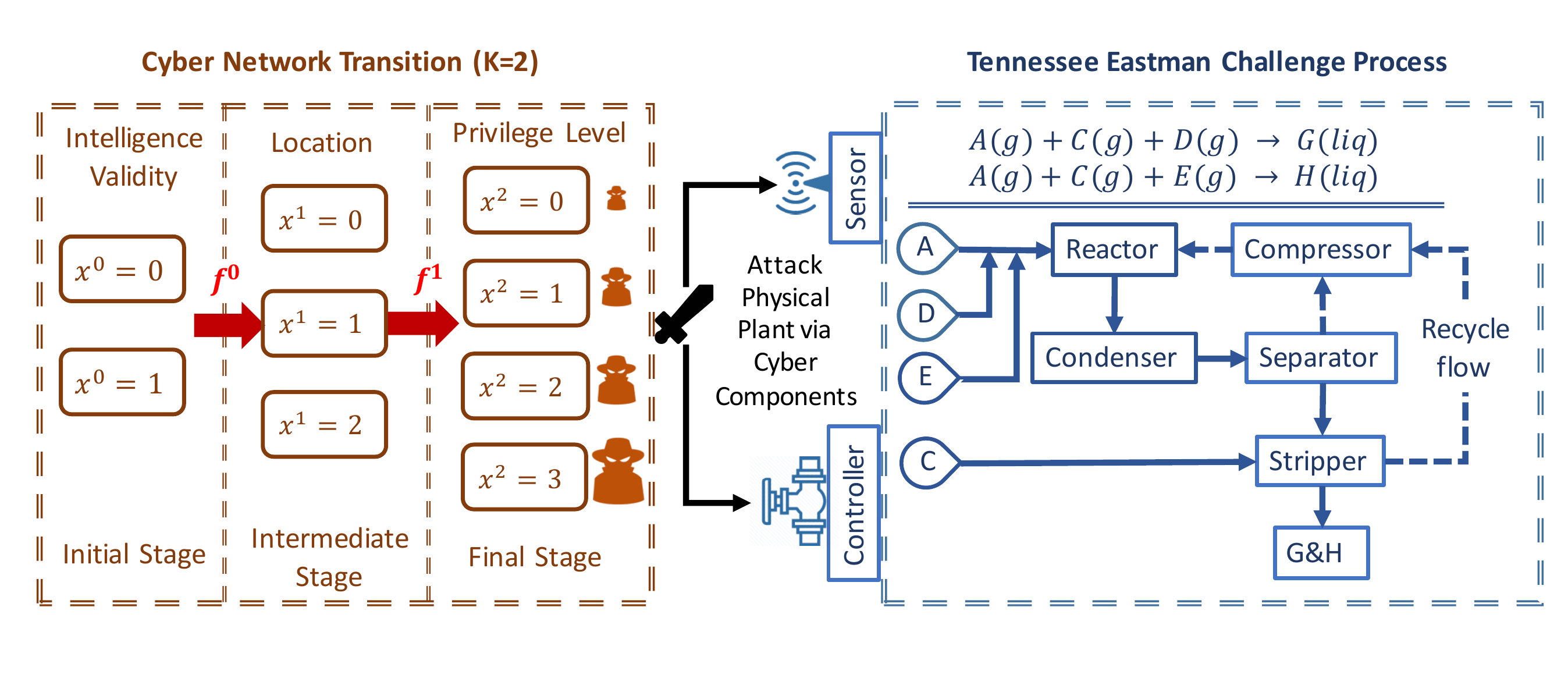}
\caption{The diagram of the cyber state transition (\textcolor{black}{denoted by the left block in orange}) and the physical attack \textcolor{black}{on Tennessee Eastman process via the compromise of the SCADA system (denoted by the right block in blue).} APTs can damage the normal industrial operation by falsifying controllers' setpoints, tampering sensor readings, and blocking communication channels to cause delays in either the control message or the sensing data. 
\label{fig: StateDiagram}}
\end{figure*}

\subsection{Initial Stage: Phishing Emails}
\textcolor{black}{We use a binary set to represent whether the reconnaissance is effectual $x^0=1$ or not $x^0=0$. Effectual reconnaissance collects essential intelligence} that can better support APTs for an initial entry through phishing emails. To penalize the adversarial exploitation of the open-source intelligence (OSINT) data, the defender can create avatars (fake personal profiles) on the social network or the company website as shown in  \cite{molok2010information}. 

\textcolor{black}{At the initial stage of interaction, a user can send emails with non-executable attachments and shortened URLs to the accounts of entry-level employees, managers, or avatars. These three action options of the user are represented by $a_2^0=0, 1, 2$, respectively. 
Non-executable files such as PDF and MS Office are widely used in organizations yet an APT attacker can exploit them to execute malicious actions on the victim's computer. 
The shortened URL is created by legitimate service providers such as \textit{Google URL shortener} yet can redirect to malicious links. 
The existing email security mechanisms are not completely effective for identifying malicious PDF files (see \cite{nissim2015detection}) and malicious links  behind  shortened URLs (see \cite{sahoo2017malicious}). 
As a supplement to technical countermeasures,  security training should be emphasized to increase employees' security awareness and  protect them from web phishing.}   
For example, after receiving suspicious links or attachments \textcolor{black}{with strange names  at unexpected times}, the entry-level employee and the manager should be aware of the potential risk and \textcolor{black}{apply extra security measures} such as a digital signature request from the sender before clicking the link or opening the attachment. 
\textcolor{black}{They should also be sufficiently alert and report immediately if a PDF does not contain the information that it claims to have. 
Then isolation can be applied to prevent the attacker from the potential lateral movement.}  
Since employees' awareness and alertness diminish over time, the security training needs to be repeated at reasonable intervals as argued in \cite{mitnick2011art}, which can be costly. 
With a limited budget, the defender can choose to educate entry-level employees, manager-level employees, or no training to avoid the prohibitive training cost $c^0$. \textcolor{black}{These three action options of the defender are represented by  $a_1^0=1,2,0$, respectively.} 
The utility matrix of the initial infection is given in Table \ref{table: multiBay-initial}. 
\textcolor{black}{If the user is legitimate, i.e., $\theta_2=\theta_2^g$, then as denoted in the blue color, he receives an immediate reward $r_1^0$ if he successfully communicates with the employee or the manager by email, but receives a substantial penalty $r_{g,f}^0<0$ if he emails the avatars because he should not contact a non-existing person. 
If the user is adversarial, i.e., $\theta_2=\theta_2^b$, then as denoted in the red color, he receives an immediate attack reward $r_2^0$ if the email receiver does not have proper security training, but an additional attack cost $r^0$ if the receiver has been trained properly.}  
The adversarial user receives a faked reward $r_{b,f}^0>0$ when contacting the avatar, yet arrives at an unfavorable state at stage $k=1$ and receives few rewards in the future stages. 
The training cost and the attack cost are both different for the primitive and the sophisticated defender, i.e., $c^0:=c_L^0 \cdot\mathbf{1}_{\{\theta_1=\theta_1^L\}}+c_H^0 \cdot \mathbf{1}_{\{\theta_1=\theta_1^H\}}$ and  $r^0:=r_L^0 \cdot \mathbf{1}_{\{\theta_1=\theta_1^L\}}+r_H^0 \cdot \mathbf{1}_{\{\theta_1=\theta_1^H\}}$. 
The sophisticated defender holds the security training with a higher frequency, which incurs a higher cost, i.e.,  $c_H^0>c_L^0$,  but is also more effective in mitigating web phishing, i.e.,  $r_H^0>r_L^0$. 

\begin{table}[]
\caption{\textcolor{black}{The expected utilities of the defender and the user at the initial stage, i.e., ${J}_1^0$ and ${J}_2^0$, respectively.} 
\label{table: multiBay-initial}}
\scalebox{1}{
\begin{tabular}{|M{1.5 cm}| M{2.3cm} | M{2.3cm}|M{2.5cm} | N  }
\hline
$\textcolor{red}{\theta_2^b}\textbf{;}\textcolor{blue}{\theta_2^g}$   & \textbf{Email Employees} & \textbf{Email Managers}  & \textbf{Email Avatars} & \\ [11pt] \hline
\textbf{No Training}           & $\textcolor{red}{(-r_2^0,r_2^0)}\textbf{;} \textcolor{blue}{(0,r_1^0)}$   & $\textcolor{red}{(-r_2^0, r_2^0)}\textbf{;}\textcolor{blue}{(0, r_1^0)}$  & $\textcolor{red}{(0,r_{b,f}^0)}\textbf{;} \textcolor{blue}{(0,r_{g,f}^0)}$ &\\ [15pt] \hline
\textbf{Train Employees}         & $\textcolor{red}{(-c^0,-r^0)}\textbf{;}\textcolor{blue}{(-c^0,r_1^0)}$  & $\textcolor{red}{(-c^0, r_2^0)}\textbf{;}\textcolor{blue}{(-c^0,r_1^0)}$  & $\textcolor{red}{(-c^0,r_{b,f}^0)} \textbf{;} \textcolor{blue}{(-c^0,r_{g,f}^0)}$ & \\  [15pt] \hline
\textbf{Train Managers}           & $\textcolor{red}{(-c^0,r_2^0)} \textbf{;}\textcolor{blue}{(-c^0,r_1^0)}$  & $\textcolor{red}{(-c^0,-r^0)}\textbf{;}\textcolor{blue}{(-c^0,r_1^0)}$  & $\textcolor{red}{(-c^0,r_{b,f}^0)} \textbf{;} \textcolor{blue}{(-c^0,r_{g,f}^0)}$ &\\ [15pt] \hline
\end{tabular}
}
\end{table}
\subsection{Intermediate Stage: Privilege Escalation}
The state  at the intermediate stage can be interpreted as the location of the user where $x^1=1$ refers to the employee's computer, $x^1=2$ refers to the manager's computer, and $x^1=0$ refers to the quarantine area. 
\textcolor{black}{After the initial access, the user operates within a process of low privilege. To access certain resources, the user needs to gain higher-level privileges. 
An attacker can utilize the process injection to execute malicious code in the address space of a live process and masquerade as legitimate programs to evade detection as shown in  \cite{processinjection}. 
A mitigation method for the defender is to prevent certain endpoint behaviors that can occur during the process injection.} 
Table \ref{table: muti-stageBay} presents this game of privilege escalation. 

The user can choose to escalate his privileges, or choose `\textit{no operation performed  (NOP)'}. \textcolor{black}{The two action options are denoted by $a_2^1=1$ and $a_2^1=0$, respectively.} 
The defender can choose to either restrict or permit an escalation, which are denoted by $a_1^1=1$ and  $a_1^1=0$, respectively. 
\textcolor{black}{
If the legitimate user escalates his privilege and the defender permits escalation, then both players obtain a reward of $r_1^1$. 
If the legitimate user escalates his privilege and the defender restricts escalation, then the efficiency reduction brings a loss of $r_1^1$ to both players. 
On the other hand, if the adversarial user escalates his privilege and the defender permits escalation, the defender receives a loss of $r_2^1$. 
If the adversarial user escalates his privilege and the defender restricts escalation, then the adversarial user has to resort to other attack techniques which lead to a higher rate of detection. Thus, the defender obtains a reward while the attacker receives an additional cost.  
We assume that the reward and the additional cost are both $r^1_L$ if the defender is primitive, and $r^1_H$ if the defender is sophisticated, i.e., $r^1=r^1_L \cdot \mathbf{1}_{\{\theta_1=\theta_1^L\}}+r^1_H \cdot \mathbf{1}_{\{\theta_1=\theta_1^H\}}$. }

\begin{table}
\caption{
\textcolor{black}{The expected utilities of the defender and the user at the intermediate stage, i.e., ${J}_1^1$ and ${J}_2^1$, respectively.}   
\label{table: muti-stageBay}}
\centering
\scalebox{1}{
\begin{tabular}{|c|c|c|N}
\hline 
$\textcolor{red}{\theta_2^b} \textbf{;}\textcolor{blue}{\theta_2^g}$  &   \textbf{NOP} & \textbf{Escalate Privilege}&\\ [7pt] \hline
\textbf{Permit Escalation}               & $\textcolor{red}{(0,0)}\textbf{;} \textcolor{blue}{(0,0)}$   & $\textcolor{red}{(-r^1_2, r^1_2)} \textbf{;} \textcolor{blue}{(r^1_1, r^1_1)}$   &\\ [10pt] \hline
\textbf{Restrict Escalation}          & $\textcolor{red}{(0,0)}\textbf{;} \textcolor{blue}{(0,0)}$  & $\textcolor{red}{(r^1,-r^1)} \textbf{;} \textcolor{blue}{(-r^1_1,-r^1_1)}$    &\\ [10pt] \hline
\end{tabular}
}
\end{table}



\subsection{Final Stage: Sensor Compromise}
\textcolor{black}{The state at the final stage represents four possible privilege levels, denoted by $x^2=\{0,1,2,3\}$, respectively. The privilege level affects the result of the physical attack at the final stage.}
The defender's and the user's actions, and the state at the intermediate stage determine the state at the final stage.
For example, if the user is at the quarantine area during the intermediate stage, then he ends up with a level-zero privilege  regardless of actions taken by the defender and himself. 
Users who take control of the manager's computer at the intermediate stage can obtain a higher privilege level than those who start from the entry-level employee's computer, yet the degree of escalation is reduced if the defender chooses to restrict escalation.

We \textcolor{black}{modify the Simulink model} in \cite{BATHELT2015309} to quantify the monetary loss of the Tennessee Eastman process under sensor compromises. 
\textcolor{black}{Our attack model of sensor compromise is presented in Section \ref{sec:attackmodel}.  A new performance metric to quantify the operational efficiency of the Tennessee Eastman process is proposed in Section \ref{sec:metric} and applied in the game matrix in Section \ref{sec:utility}. }

\subsubsection{Performance Metric}
\label{sec:metric}
The Tennessee Eastman process involves two irreversible reactions to produce two liquid (liq) products $G,H$ from four gaseous (g) reactants $A,C,D,E$ \textcolor{black}{as shown in the right block of Fig. \ref{fig: StateDiagram}.} 
The control objective is to maintain a desired production rate as well as quality while stabilizing the whole system under the Gaussian noise to avoid violating safety constraints such as a high reactor pressure, a high reactor temperature, and a high/low separator/stripper liquid level. 
Previous studies on the security of the Tennessee Eastman process have mostly focused on \textcolor{black}{how an attacker can cause} the shortest shutdown time (see   \cite{krotofil2013resilience}), or a \textcolor{black}{serious violation} of a setpoint, e.g., the reactor pressure exceeds $3,000$ kpa (see  \cite{cardenas2011attacks}). 
\textcolor{black}{These attacks successfully cause the shutdown of the plant and a few days of shutdowns can incur a considerable financial loss. However, the shutdown also discloses the attack and leads to an immediate patch and a defense strategy update. 
Thus, it becomes harder for the same kind of attacks to succeed after the plant recovers from the shutdown.} 

In our APT scenario, the attacker \textcolor{black}{aims to stealthily} decrease the operational efficiency of the plant, i.e., deviate the normal operation state of the plant without triggering the safety alarm or shutting down the plant. 
\textcolor{black}{By compromising the SCADA system and generating fraudulent sensor readings, the attacker can stealthily make the plant operates at a non-optimal state with reduced utilities. 
} 
\textcolor{black}{The following economic metrics affect the operational utility of the Tennessee Eastman process:} 
\begin{itemize}
\item Hourly \textcolor{black}{operating cost} $C_o$ with the unit  $(\$/h)$ is taken as the sum of purge costs, product stream costs, compressor costs, and stripper steam costs.

\item  Production rate $R_p$ with the unit $(m^3/h)$ is the volume of total products per hour. 
\item Quality of products $Q_p$ with the unit $ (G \ { mole} \%)$, is the percentage of $G$ among total products. 
\item $P_G$ with the unit $(\$/m^3)$ is the price of product $G$. 
\end{itemize}

\textcolor{black}{We propose a new performance metric $U_{TE}$, the \textit{per-hour utility} to quantify the operational efficiency of the Tennessee Eastman process as follows:}
\begin{equation}
\label{eq: Tennessee Eastman utility}
U_{TE}=R_p\times Q_p \times P_G -C_o.
\end{equation}

\subsubsection{Attack Model}
\label{sec:attackmodel}
An attack model \textcolor{black}{is characterized by} two separate parts, \textit{information} and \textit{capacity}. 
First, the information available to the attacker such as \textcolor{black}{readings of different sensors} can affect the performance of the attack \textcolor{black}{differently}.   
For example, observing the input rate of the raw material in the Tennessee Eastman process is less beneficial for the attacker than the direct measurements of  $P_G, R_p, Q_p, C_o$ that affect the utility metric in \eqref{eq: Tennessee Eastman utility}. 
Second, attackers can have different capacities in accessing and revising controllers and sensors. 
An attacker may change the parameters of the \textcolor{black}{proportional-integral-derivative} controller,  directly falsify the controller output, or indirectly deviate the setpoint by tampering, blocking or delaying sensor readings. 

In this experiment, we assume a \textcolor{black}{reading} manipulation of sensor XMEAS$(40)$ and XMEAS$(17)$ in loop $8$ and loop $13$ of Tennessee Eastman process (see  \cite{ricker1996decentralized}), respectively.  
Sensor XMEAS$(40)$ measures the composition of component $G$ 
and sensor XMEAS$(17)$ measures the stripper underflow. 
A higher privilege state $x^2\in \{0,1,2,3\}$ means that the user can access more sensors for a longer time, \textcolor{black}{which results in a larger loss and thus a smaller utility of $r_1^2(x^2)$ to the defender if the user is adversarial.} 
  Fig.  \ref{fig: Tennessee Eastman} shows the variation of $U_{TE}$ versus the simulation time under four different privilege states.  
  \textcolor{black}{We use the time average of these utilities to obtain the normal operational utility $r_4^2$ and \textcolor{black}{compromised} utilities $r_1^2(x^2)$ under four different privilege states $x^2\in \{0,1,2,3\}$.} 
\begin{figure}
    \centering
    \begin{subfigure}[]{0.89\textwidth}
        \centering
        \includegraphics[width=1 \textwidth]{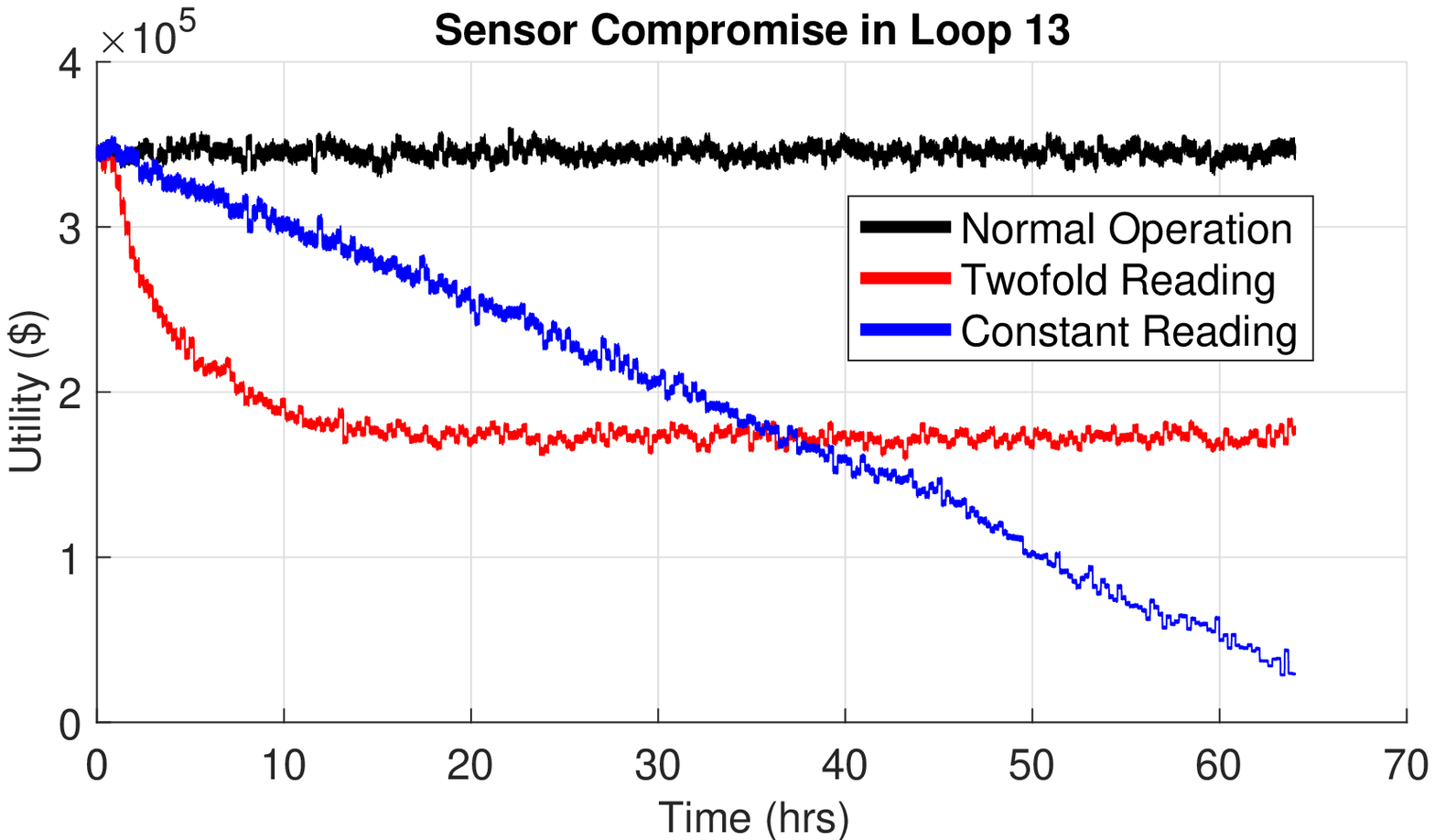}
    \end{subfigure}%
    
    \begin{subfigure}[]{0.89\textwidth}
        \centering
        \includegraphics[width=1 \textwidth]{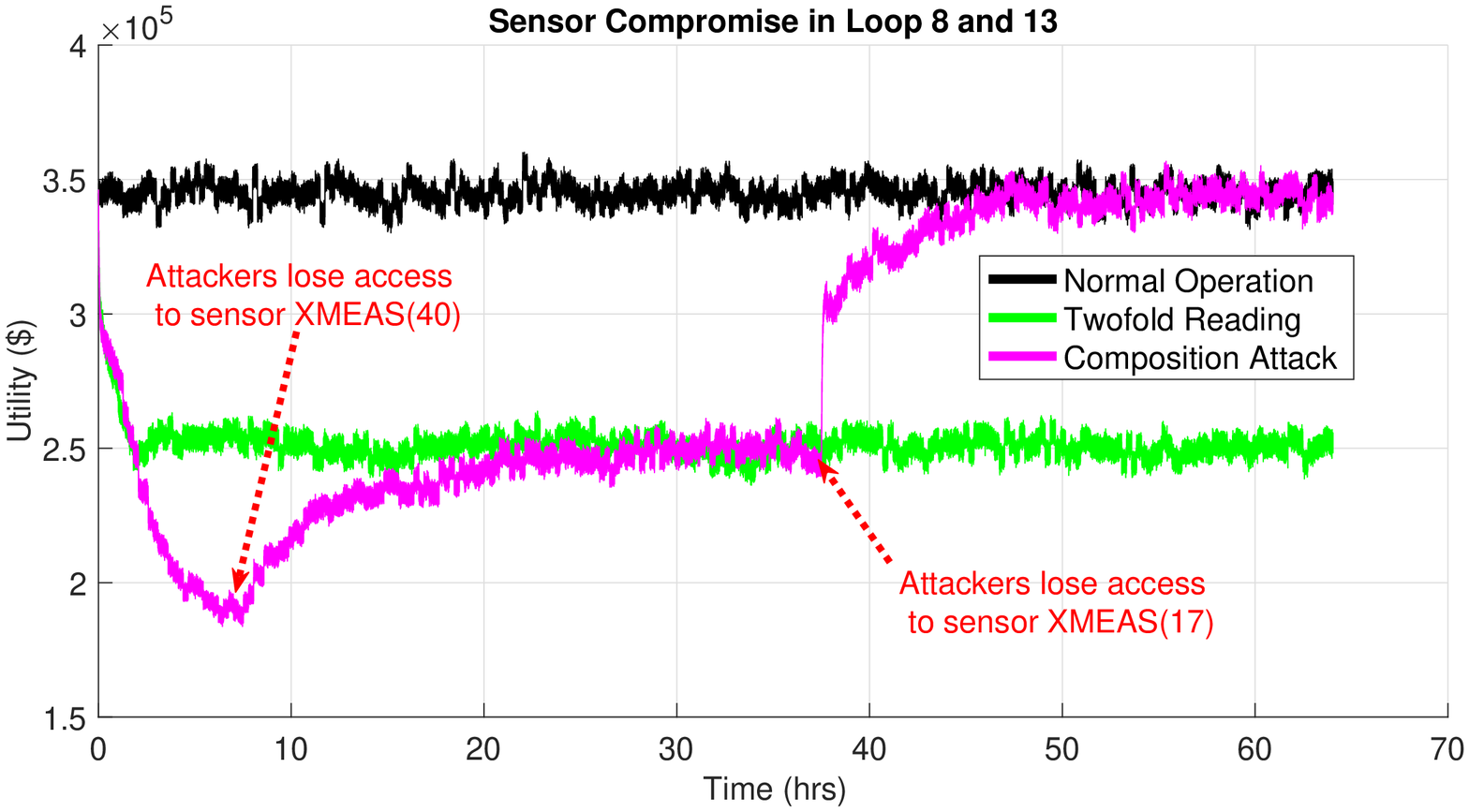}
    \end{subfigure}
    \caption{\label{fig: Tennessee Eastman}
    The economic impact of sensor compromise in the Tennessee Eastman process. 
    The black line represents the utility of Tennessee Eastman process under the normal operation while the other four lines represent the utility of Tennessee Eastman process under \textcolor{black}{attacks with four possible privilege levels.} 
    We use the time average of these utilities to obtain the normal operational utility $r_4^2$ and \textcolor{black}{compromised} utilities $r_1^2(x^2), \forall x^2\in \{0,1,2,3\}$, under four different \textcolor{black}{states of privilege levels} in Table \ref{Table: FinalMatrix}. 
 }
\end{figure}
The attacker compromises the sensor and generates fraudulent readings. The fraudulent reading can be a constant, denoted by the blue line, or a double of the real readings, denoted by the red or green lines. 
The pink line represents a composition attack 
with a limited control time. 
Initially, the attacker manages to compromise both sensors by doubling their readings. 
After the attacker loses access to XMEAS$(40)$ at the $6^{th}$ hour, the system is sufficiently resilient to recover partially in about $16$ hours and achieve the same level of utility as the single attack in green.  When the attacker also loses access to XMEAS$(17)$ at the $36^{th}$ hour, the utility goes back to normal in about $13$ hours.   

\subsubsection{Utility Matrix}
\label{sec:utility}
\textcolor{black}{Attacks against SCADA system can apply command injection attacks to inject false control and compromise sensor readings as shown in \cite{morris2013industrial}. 
Encryption can be introduced to conceal these malicious commands. 
However, a legitimate user may also encrypt his communication with the sensor to avoid eavesdropping and enhance privacy.}

Therefore, at the final stage, the user has two options, sends commands to the sensor with or without encryption, which are denoted by $a_2^2=1$ and $a_2^2=0$, respectively. 
The defender chooses to apply either a complete or selective monitoring, denoted by $a_1^2=1$ and $a_1^2=0$, respectively. 
\textcolor{black}{
The complete monitoring stores all sets of communication data and analyzes them elaborately to identify malicious commands despite encryption. 
The selective monitoring cannot identify malicious commands if they are encrypted. 
The implementation of the complete monitoring incurs an additional cost $c^2$ compared to the selective one. 
The last-stage utility matrix of both players is defined in Table \ref{Table: FinalMatrix}. 
If the user is legitimate, as denoted in blue, both the defender and the user can receive a reward of $r^4$ when the Tennessee Eastman process operates normally. Legitimate users further receive a utility reduction of $r^4/2$ for the potential privacy loss if they choose unencrypted commands. 
For adversarial users, they send malicious commands only when the communication is encrypted to evade detection. Thus, if they choose not to encrypt the communication, they receive $0$ utility and the defender receives a reward of $r^4$ for the normal operation. However, if they choose to send encrypted malicious commands, both players' rewards depend on whether the defender chooses the selective or complete monitoring. If the defender chooses the selective monitoring, then the adversarial user can successfully compromise the sensor, which results in a reduced utility of $r_1^2(x^2)$. In the meantime, the attacker benefits from the reward reduction of $r_4^2-r_1^2(x^2)$. If the defender chooses the complete monitoring, then the adversarial user suffers a loss of $r^2$ for being detected.} 
The detection reward and the implementation cost for two types of defenders are $r_L^2,r_H^2$ and $ c_L^2,c_H^2$, respectively. 
Let $r^2:=r_L^2 \cdot \mathbf{1}_{\{\theta_1=\theta_1^L\}}+r_H^2 \cdot \mathbf{1}_{\{\theta_1=\theta_1^H\}}$ and $c^2:=c_L^2 \cdot \mathbf{1}_{\{\theta_1=\theta_1^L\}}+c_H^2 \cdot \mathbf{1}_{\{\theta_1=\theta_1^H\}}$.

\begin{table}
\caption{\label{Table: FinalMatrix}
\textcolor{black}{The expected utilities of the defender and the user at the final stage, i.e., ${J}_1^2$ and ${J}_2^2$, respectively.} 
}
\centering
\scalebox{1}{
\begin{tabular}{|c|c|c|N}
\hline
$\textcolor{red}{\theta_2^b} \textbf{;} \textcolor{blue}{\theta_2^g}$& \textbf{Unencrypted Command (UC)} & \textbf{Encrypted Command (EC)} &\\ [7pt] \hline
\textbf{Selective Monitoring (SM)}               & $\textcolor{red}{(r_4^2,0)}\textbf{;}\textcolor{blue}{(r_4^2,r_4^2/2)}$   & $\textcolor{red}{(r_1^2(x^2), r_4^2-r_1^2(x^2))}\textbf{;}\textcolor{blue}{(r_4^2 , r_4^2)}$    &\\ [10pt] \hline
\textbf{Complete Monitoring (CM)}         & $\textcolor{red}{(r_4^2-c^2,0)}\textbf{;} \textcolor{blue}{(r_4^2-c^2,r_4^2/2)}$  & $\textcolor{red}{(r^2-c^2,-r^2)} \textbf{;} \textcolor{blue}{(r^2_4-c^2,r^2_4)}$   &\\ [10pt] \hline
\end{tabular}
}
\end{table}

\section{Computation Results}
\label{sec:result}
\textcolor{black}{In this section, we apply the algorithms introduced in Section \ref{sec:computation}  to compute both players' strategies and utilities at the equilibrium. 
We implement our algorithms in MATLAB and use YALMIP (see \cite{Lofberg2004}) as the interface to call external solvers such as BARON (see \cite{ts}) to solve the optimization problems.} 
\textcolor{black}{We present elaborate results from the concrete case study and provide meaningful insights of the proactive cross-layer defense against multi-stage APT attacks that are stealthy and deceptive.}  

For the static Bayesian game at the final stage in Section \ref{sec:Finalstage}, 
\textcolor{black}{we focus on illustrating how two players' private types affect their policies and utilities under different information structures. 
We further apply sensitivity analysis to show how the value of the key parameter affects the defender's and the attacker's utilities.}  
For the multi-stage Bayesian game in \ref{sec: multi-stage and PBNE}, \textcolor{black}{we focus on the dynamic of the belief update and state transition under the interaction of the stealthy attacker and the proactive defender.  
Moreover, we investigate how the adversarial and defensive deception, and how the initial state can affect the stage utility and the cumulative utility of the user and the defender.} 

\subsection{Final Stage and SBNE}
\label{sec:Finalstage}
\begin{figure}
    \centering
    \begin{subfigure}[t]{0.48\textwidth}
        \centering
        \includegraphics[width=1 \textwidth]{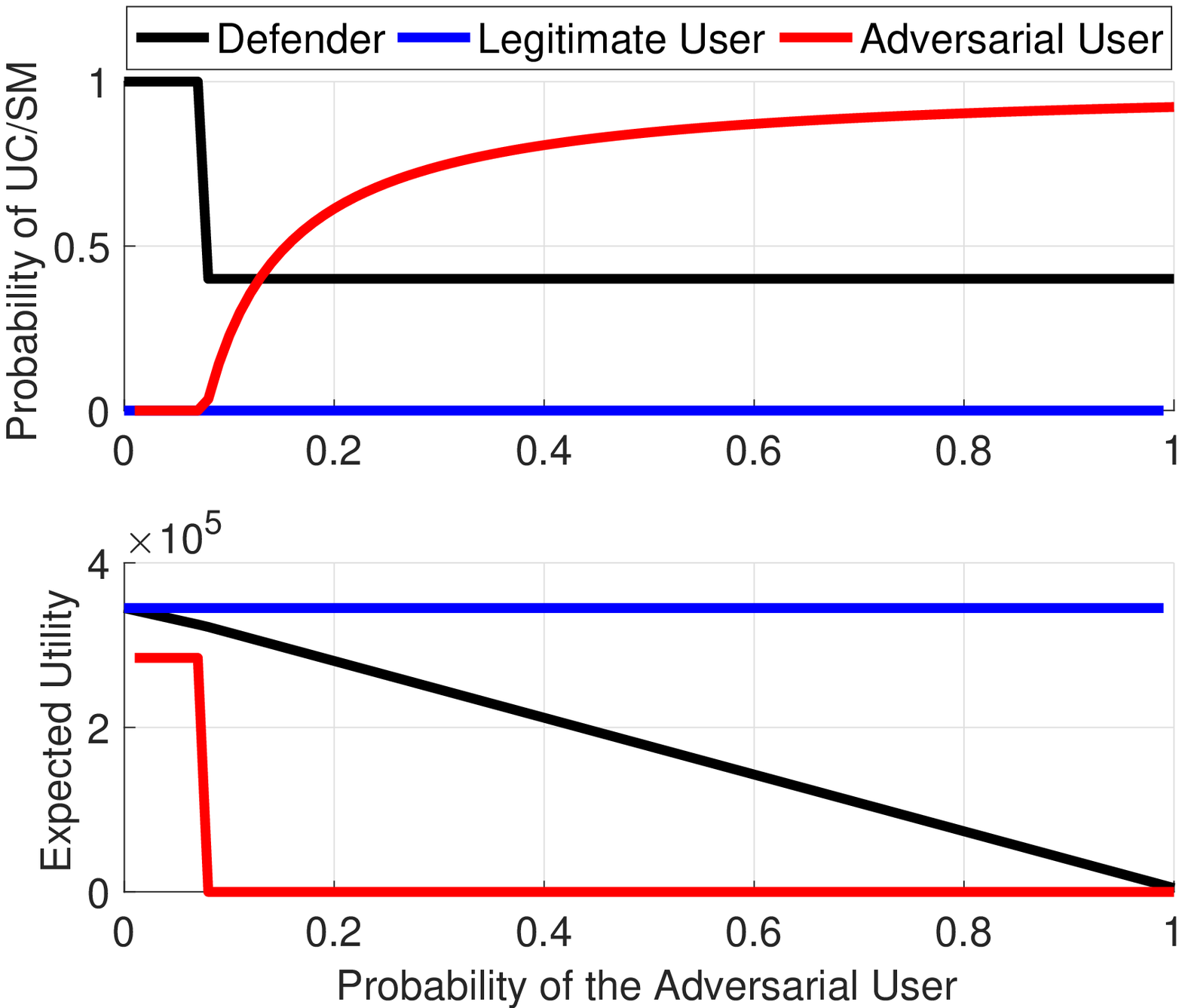}
        \subcaption{The user knows that the defender is primitive, yet the defender only knows the probability of the user being adversarial. 
        \label{fig: FinalstageBlief} }
    \end{subfigure}%
    \hfill 
    \begin{subfigure}[t]{0.48\textwidth}
        \centering
        \includegraphics[width=1 \textwidth]{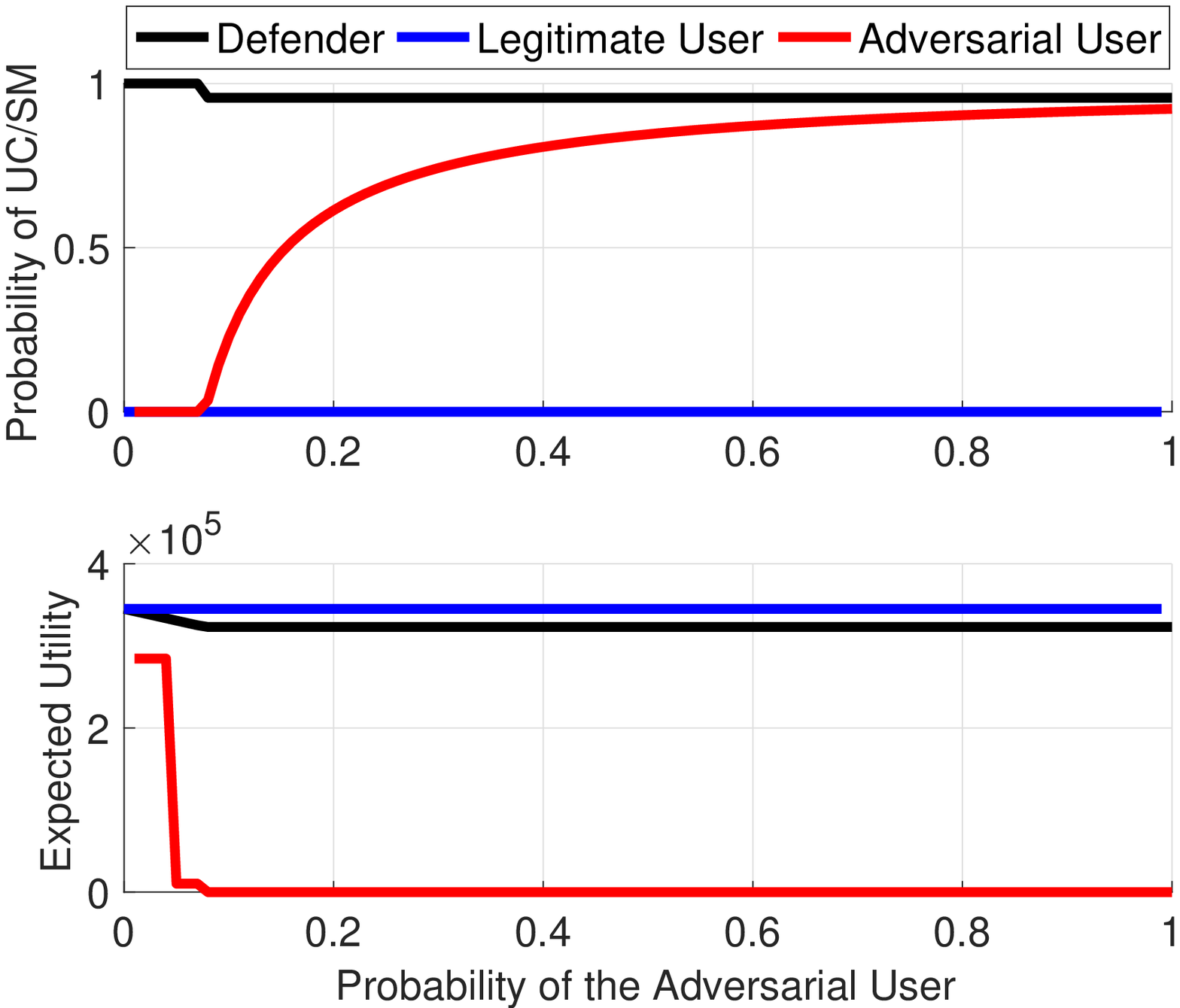}
        \subcaption{\label{fig: Finalstagebelief3revise}
    Both players' types are private, and each player only knows the probability of the other player's type. }
    \end{subfigure}
    \caption{The SBNE strategy and the expected utility of  the primitive defender and the user who is either legitimate or adversarial. The $x$-axis represents the probability of the user being adversarial. \textcolor{black}{The $y$-axis of the upper figure represents the probability of either the user taking action `\textit{selective monitoring (SM)}' or the defender taking action `\textit{unencrypted command (UC)}'.}
 }
\end{figure}

 
 \begin{figure}
 \centering
  \includegraphics[width=0.65  \textwidth]{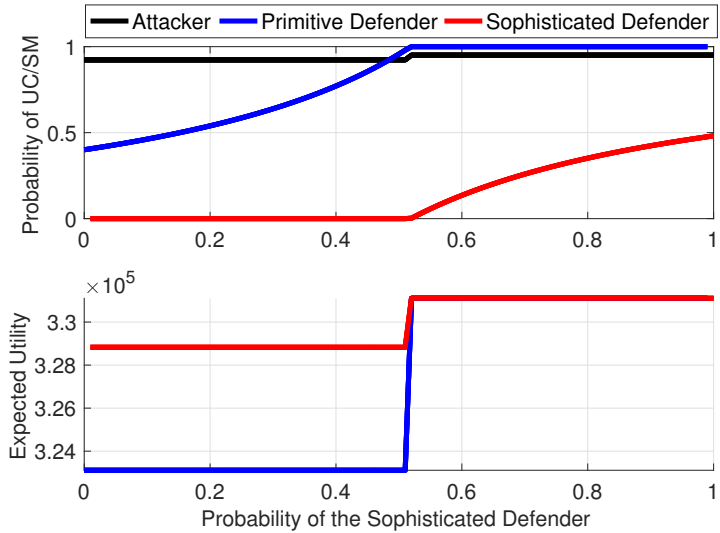}
\caption{
The SBNE strategy and the expected utility of the adversarial user and the defender who is either primitive or sophisticated. 
The defender knows that the user is adversarial while the adversarial user only knows the probability of the defender being primitive. 
The $x$-axis represents the probability of the defender being sophisticated.
\textcolor{black}{The $y$-axis of the upper figure represents the probability of either the user taking action `\textit{selective monitoring (SM)}' or the defender taking action `\textit{unencrypted command (UC)}'.}
 \label{fig: Finalstagebelief2revise}}
 \end{figure}
 
Players' beliefs affect their policies and expected utilities at the final stage. We discuss three different scenarios as follows. 
In Fig \ref{fig: FinalstageBlief}, the defender does not know the user's type. 
In Fig. \ref{fig: Finalstagebelief2revise}, the user does not know the defender's type. 
In  Fig. \ref{fig: Finalstagebelief3revise}, both the user and the defender do not know the other's type. 
In all three scenarios, the $x$-axis represents the belief of either the user or the defender. \textcolor{black}{The $y$-axis of the upper figure represents the probability of either the user taking action `\textit{selective monitoring (SM)}' or the defender taking action `\textit{unencrypted command (UC)}'.}
Fig. \ref{fig: FinalstageBlief} shows the following trends as the user becomes more likely to be adversarial. 
First, two black lines show that the expected utility of the defender decreases and the defender is more inclined to apply action `\textit{complete monitoring}' after her belief exceeds a threshold. 
Second, two red lines show that the adversarial user takes action `\textit{unencrpted command}' with a higher probability and only gains a reward when the probability of adversarial users is sufficiently small. 
Thus, we conclude that when the probability of the adversarial user increases, the defender tends to invest more in cyber defense so that the attacker behaves more conservatively and inflicts fewer losses. 
Third, the two blue lines show that the legitimate user always chooses `\textit{encrypt command}' and receives a constant utility, which indicates that the proactive defense does not affect the behavior and the utility of legitimate users at this stage. 



Fig. \ref{fig: Finalstagebelief2revise} shows that the defender benefits from introducing defensive deception. 
When the defender becomes more likely to a sophisticated one, both types of defenders can have a higher probability to apply the selective monitoring and save the extra surveillance cost of the complete monitoring. 
The attacker with incomplete information has a threshold policy and switches to a lower attacking probability after reaching the threshold of $0.5$ as shown in the black line. 
When the probability goes beyond the threshold, the primitive defender can pretend to be a sophisticated one and take action `\textit{selective monitoring}'. 
Meanwhile, a sophisticated defender can reduce the security effort and take  action 
`\textit{selective monitoring}'
with a higher probability since the attacker becomes more cautious in taking adversarial actions after identifying the defender as more likely to be sophisticated. 
It is also observed that the sophisticated defender receives a higher payoff before the attacker's belief reaches the $0.5$ threshold. After the belief reaches the threshold, the attacker is threatened to take less aggressive actions, and both types of defenders share the same payoff.

Finally, we consider the double-sided incomplete information where both players' types are private information, and each player only has the belief of the other player's type. 
Compared with the defender in Fig. \ref{fig: FinalstageBlief} who takes action `\textit{selective monitoring}' with a probability less than $0.5$ and receives a decreasing expected payoff, the defender in Fig. \ref{fig: Finalstagebelief3revise} can take `\textit{selective monitoring}' with a probability closed to $1$ and receive a constant payoff in expectation after the user's belief exceeds the threshold. 
Thus, the defender can spare defense efforts and mitigate risks by
introducing uncertainties on her type as a countermeasure to the adversarial deception. 
 


\subsubsection{Sensitivity Analysis}
\label{sec: sensitiveAna}

\begin{figure}
    \centering
    \begin{subfigure}[]{0.5\textwidth}
        \centering
        \includegraphics[width=1 \textwidth]{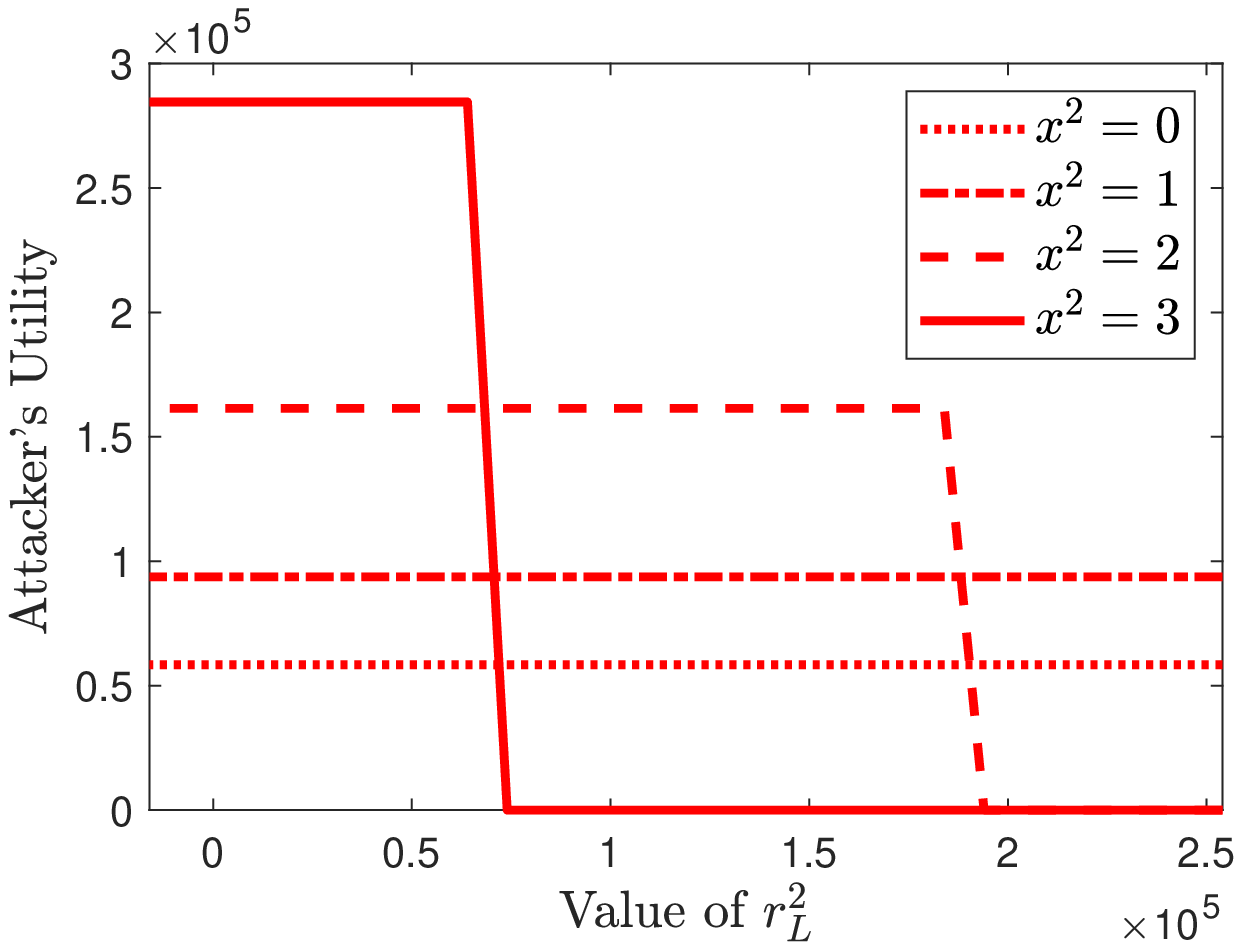}
    \end{subfigure}%
    \begin{subfigure}[]{0.5\textwidth}
        \centering
        \includegraphics[width=1 \textwidth]{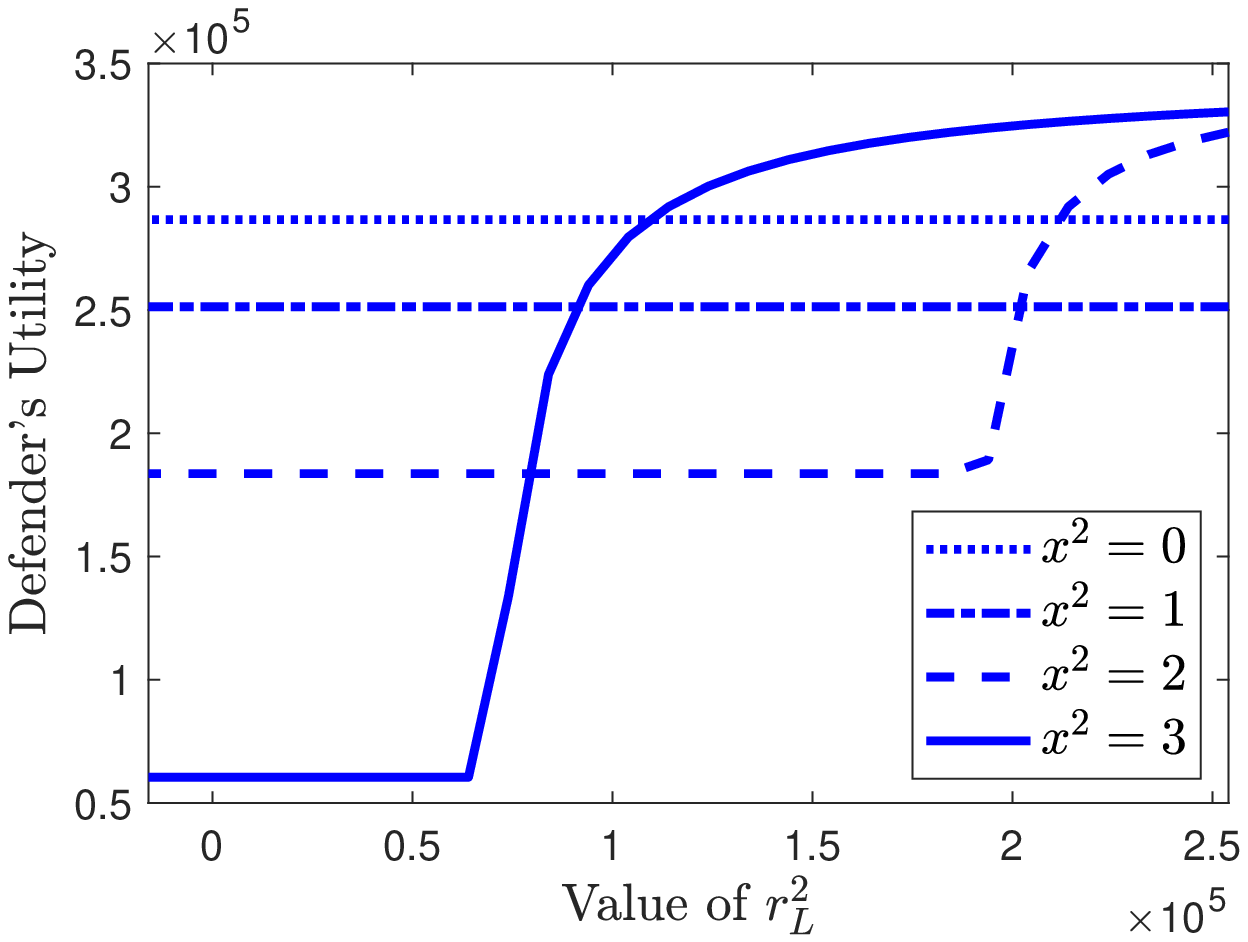}
    \end{subfigure}
   \caption{
Utilities of the primitive defender and the attacker versus the value of $r_L^2$ under different states $x^2\in \{0,1,2,3\}$. 
\label{fig: sensitivity}
}
\end{figure}

As shown in Fig. \ref{fig: sensitivity}, if the value of the penalty $r^2_L$ is close to $0$, i.e., the defense at the final stage is ineffective, then an arrival at state $x^2=3$, the highest privilege level can significantly increase the attacker's payoff and cause the most damage to the defender. 
As more effective defensive methods are employed at the final stage, i.e., the value of $r^2_L$ increases, the attacker becomes more conservative and strategic in taking adversarial behaviors. Then, the state with the highest privilege level may not be the most favorable state for the attacker. 

\begin{figure}
\centering
  \includegraphics[width=0.65\textwidth]{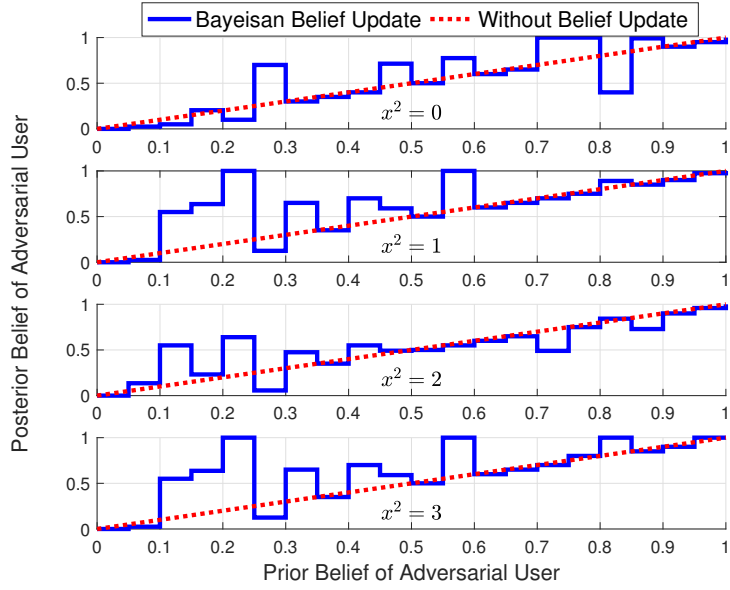}
\caption{
\label{fig: initialbelief}
The defender's prior and posterior beliefs of the user being adversarial. 
}
\end{figure}
\begin{figure}
\centering
 \includegraphics[width=0.65 \textwidth]{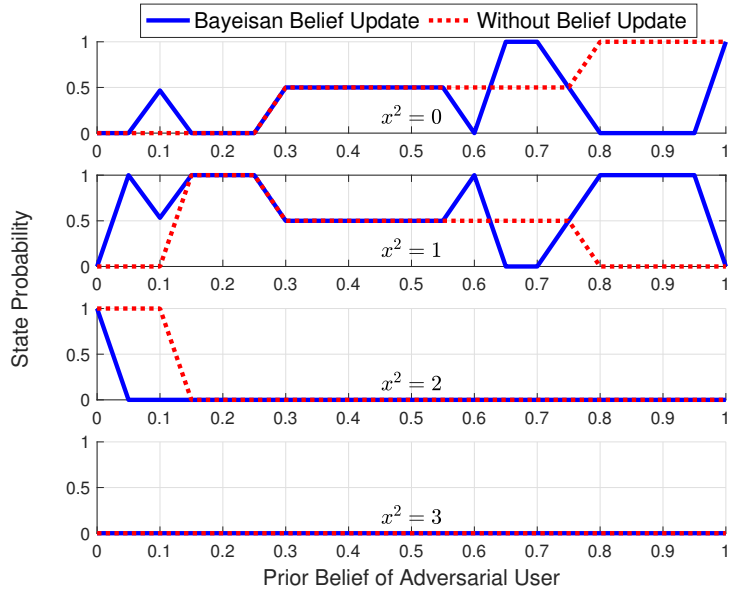}
\caption{
The probability of different states $x^2 \in\{0,1,2,3\}$. 
\label{fig: initial_statedistri}
}
\end{figure}
\begin{figure}
\centering
 \includegraphics[width=0.65 \textwidth]{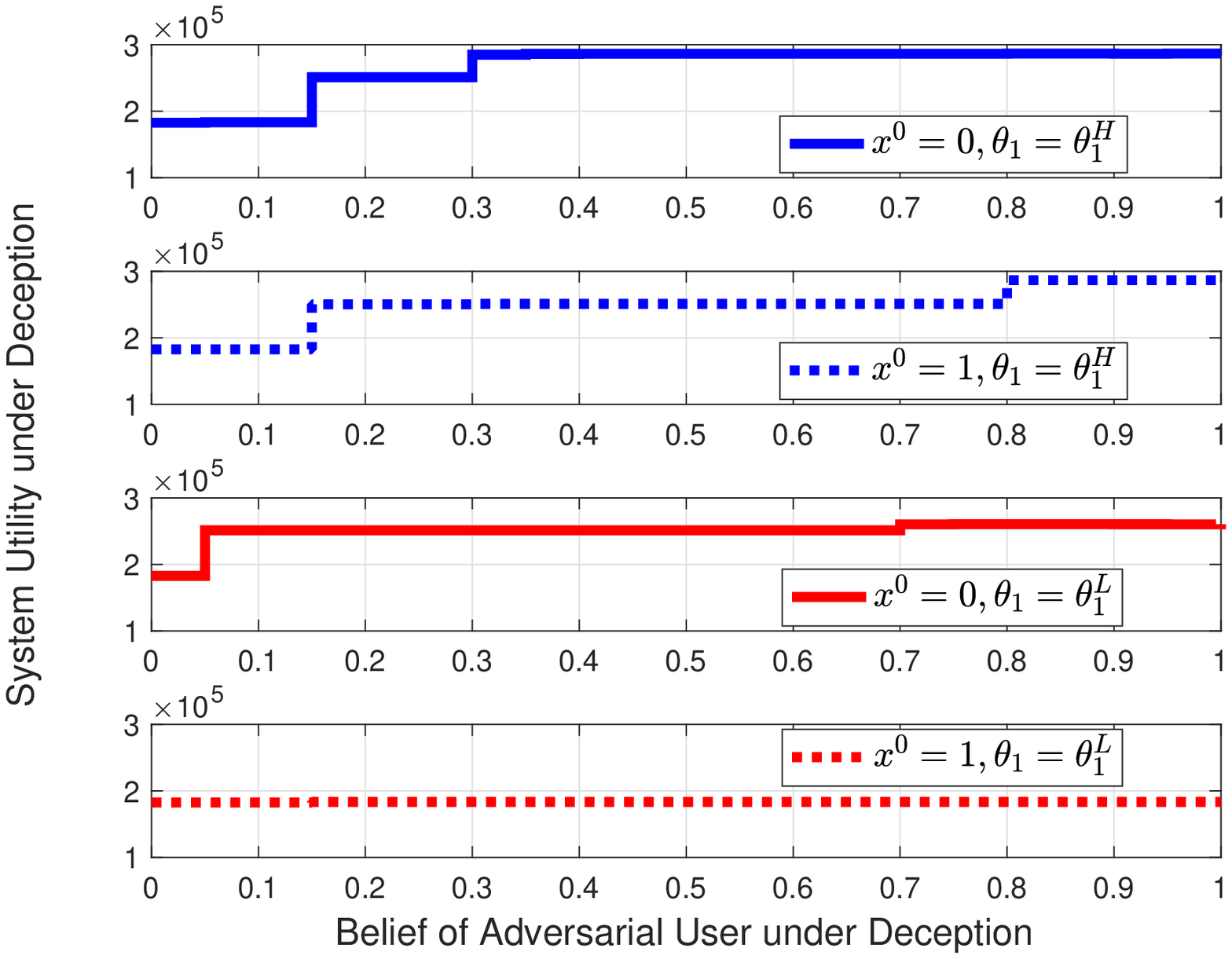}
\caption{
The defender's utility under deceived beliefs.  
\label{fig: UtilityDeception}
}
\end{figure}

\subsection{Multi-stage and PBNE}
\label{sec: multi-stage and PBNE}
We show in Fig. \ref{fig: initialbelief} that the Bayesian belief update leads to a more accurate estimate of users' types. 
Without the belief update, the posterior belief is the same as the prior belief in red and is used as the baseline. 
As the prior belief increases in the $x$-axis, the posterior belief after the Bayesian update also increases in blue. 
The blue line is in general above the red line, which means that with the Bayesian update, the defender's belief becomes closer to the right type. 
Also, we find that the belief update is the most effective when an inaccurate prior belief is used as it corrects the erroneous belief significantly. 


%

In Fig. \ref{fig: initial_statedistri}, we show that the proactive defense, i.e., defensive methods in intermediate stages can affect the state transition and reduce the probability of attackers reaching states that can result in huge damage at the final stage. 
As the prior belief of the user being adversarial increases, the attacker is more likely to arrive at state $x^2=0$ and  $x^2=1$, and reduce the probability of visiting  $x^2=2$ and  $x^2=3$. 


\subsubsection{Adversarial and Defensive Deception}
\label{sec:deception}
Fig. \ref{fig: UtilityDeception} investigates the adversarial deception where the attacker takes full control of the \textcolor{black}{defense} system and manipulates the defender's belief. 
As shown in the figure, the defender's utilities all increase when the belief under the deception approaches the correct belief that the user is adversarial. 
Also, the increase is stair-wise, i.e., the defender only alternates her policy when the manipulated belief is beyond certain thresholds. 
Under the same manipulated belief, a sophisticated defender benefits no less than a primitive one. \textcolor{black}{The defender receives a lower payoff when the reconnaissance provides effectual intelligence.} 

Incapable of \textcolor{black}{revealing} the adversarial deception \textcolor{black}{completely},  the defender  can alternatively introduce defensive deceptions, e.g., 
a primitive defender can disguise himself as a sophisticated one to confuse the attacker. Defensive deceptions introduce uncertainties to attackers, increase their costs, and increase the defender's utility. 
\begin{figure}[]
\centering
\includegraphics[width=0.8 \textwidth]{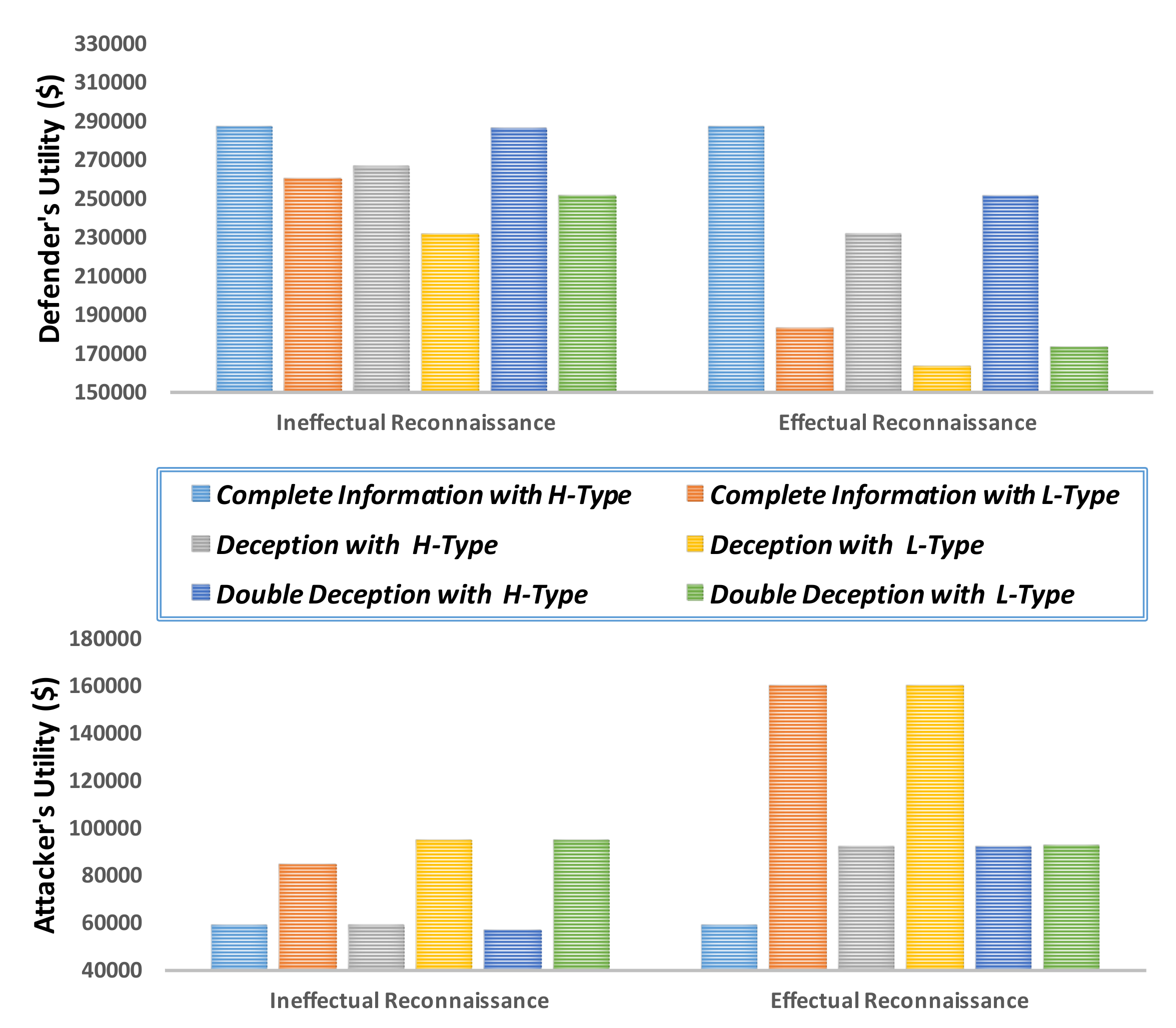}
\caption{
The cumulative utilities of the attacker and the defender under the complete information, the adversarial deception, and the defensive deception.  
In the legend, the left three represent the utilities for a sophisticated defender and the right three represent the ones for a primitive defender. 
 \label{fig: v12compare}}
\end{figure}
Fig. \ref{fig: v12compare} investigates the defender's and the attacker's utilities under three different scenarios. 
The complete information refers to the scenario where both players know the other player's type. 
The deception with the $H$-type or the $L$-type means that the attacker knows the defender's type to be sophisticated or primitive, respectively,  yet the defender has no information about the user's type. 
The double-sided deception indicates that both players do not know the other player's type. 
The results from Fig. \ref{fig: v12compare} are summarized as follows. 
First, the sophisticated defender's payoffs can increase as much as $56\%$ than those of the primitive defender. 
 Also, a prevention of effectual reconnaissance increases the defender's utility by as much as $41\%$ and reduces the attacker's utility  by as much as $38\%$. 
Second, the defender and the attacker receive the highest and the lowest payoff, respectively, under the complete information.
When the attacker introduces deceptions over his type, the attacker's utility increases and the defender's utility decreases. 
Third, when the defender adopts defensive deceptions to introduce double-sided incomplete information, we find that the decrease of the sophisticated defender's utilities is reduced by at most $64 \%$, i.e., changes from $\$55,570$ to $\$35,570$ when the reconnaissance is effectual.   
The double-sided incomplete information also brings lower utilities to the attacker than the 
one-sided adversarial deception. 
However, the defender's utility under the double-sided deception is still less than the complete information case, which concludes that acquiring complete information of the adversarial user is the most effective defense. However, if the complete information cannot be obtained, the defender can mitigate her loss by introducing defensive deceptions. 


\section{Discussions and Conclusions}
\label{sec:conclusion}
Advanced Persistent Threats (APTs) are emerging security challenges for cyber-physical systems as the attacker can stealthily enter, persistently stay in, and strategically interact with the system. 
In this work, we have developed a game-theoretic framework to design proactive and cross-layer defenses for cyber-physical systems in a holistic manner. 
Dynamic games of incomplete information have been used to capture the long-term interaction between users  and defenders who have private information unknown to the other player. 
Each player forms a belief on the unknowns and uses the Bayesian update to learn the private information and reduce uncertainty. The analysis of the Perfect Bayesian Nash Equilibrium (PBNE) has provided the defender with an effective countermeasure against the stealthy strategic attacks at multiple stages.  
To compute the PBNE of the dynamic games, we have proposed a nested algorithm that iteratively alternates between the forward belief update and the backward policy computation. 
The algorithm has been shown to quickly converge to the $\varepsilon$-PBNE that yields a consistent pair of beliefs and policies. 


Using \textcolor{black}{the} Tennessee Eastman  process as a case study of industrial control systems, we have shown that the proactive multi-stage defense in cyber networks can successfully mitigate the risk of physical attacks without reducing the payoffs of legitimate users. 
In particular, experiment results show that a sophisticated defender receives a payoff up to $56\%$ higher than a primitive defender does. Also, it has been illustrated that \textcolor{black}{by preventing effectual reconnaissance}, the defender increases her utility and reduces the attacker's utility by at most $41\%$ and $38\%$, respectively. 
On one hand, the attacker receives a higher payoff after introducing \textcolor{black}{the adversarial deception as it increases the defender's uncertainties on the user's type.} 
On the other hand, by creating uncertainties for attackers, the defender can successfully threaten them to take more conservative behaviors and become less motivated to launch attacks. 
It has been shown that the defender can significantly benefit from the mitigation of attack losses when he adopts defensive deceptions. 


\textcolor{black}{The main challenge of our approach is to identify the utility and feasible actions of defenders and users at each stage. 
One future direction to reduce the complexity of the model description is to develop mechanisms that can automate the synthesis of verifiably correct game-theoretic models. It would alleviate the workload of the system defender and operator.
Nevertheless, game theory provides a quantitative and explainable framework to design the proactive defensive response under uncertainty compared to rule-based and machine-learning-based defense methods, respectively.
Besides, the rule-based defense is static, thus an attack can circumvent it through sufficient effort. Machine learning methods require a lot of labeled data sets which may be hard to obtain in the APT scenario.
Second, we have proposed the belief to quantify the uncertainty which results from players' private types. The belief is continuously updated to reduce uncertainties and provide a probabilistic detection system as a byproduct of the APT response design.
Third, our approach enables the defender to evaluate the multi-stage impact of her defense strategies to both legitimate and adversarial users when adversarial and defensive deceptions present at the same time. 
Based on the evaluation, defenders can further find revised countermeasures and design new game rules to achieve a better tradeoff between security and usability.
}
\textcolor{black}{Our model can be broadly applied to scenarios in artificial intelligence, economy, and social science where multi-stage interactions occur between multiple agents with incomplete information. Multi-sided non-binary types can be defined based on the scenario, and our iteration algorithm of the forward belief update and the backward policy computation can be extended for efficient computations of the perfect Bayesian Nash equilibrium.}
The future work would extend the framework to an $N$-person game to characterize the simultaneous interactions among multiple users and model composition attacks. 
We would also consider scenarios where players' actions and the system state are partially observable.

%
%

\bibliographystyle{cas-model2-names} 

\bibliography{APT}

\begin{thebibliography}{44}
\expandafter\ifx\csname natexlab\endcsname\relax\def\natexlab#1{#1}\fi
\providecommand{\url}[1]{\texttt{#1}}
\providecommand{\href}[2]{#2}
\providecommand{\path}[1]{#1}
\providecommand{\DOIprefix}{doi:}
\providecommand{\ArXivprefix}{arXiv:}
\providecommand{\URLprefix}{URL: }
\providecommand{\Pubmedprefix}{pmid:}
\providecommand{\doi}[1]{\href{http://dx.doi.org/#1}{\path{#1}}}
\providecommand{\Pubmed}[1]{\href{pmid:#1}{\path{#1}}}
\providecommand{\bibinfo}[2]{#2}
\ifx\xfnm\relax \def\xfnm[#1]{\unskip,\space#1}\fi
\bibitem[{Bathelt et~al.(2015)Bathelt, Ricker and Jelali}]{BATHELT2015309}
\bibinfo{author}{Bathelt, A.}, \bibinfo{author}{Ricker, N.L.},
  \bibinfo{author}{Jelali, M.}, \bibinfo{year}{2015}.
\newblock \bibinfo{title}{Revision of the {T}ennessee {E}astman process model}.
\newblock \bibinfo{journal}{IFAC-PapersOnLine} \bibinfo{volume}{48},
  \bibinfo{pages}{309 -- 314}.
\newblock \DOIprefix\doi{https://doi.org/10.1016/j.ifacol.2015.08.199}.
  \bibinfo{note}{9th IFAC Symposium on Advanced Control of Chemical Processes
  ADCHEM 2015}.
\bibitem[{C{\'a}rdenas et~al.(2011)C{\'a}rdenas, Amin, Lin, Huang, Huang and
  Sastry}]{cardenas2011attacks}
\bibinfo{author}{C{\'a}rdenas, A.A.}, \bibinfo{author}{Amin, S.},
  \bibinfo{author}{Lin, Z.S.}, \bibinfo{author}{Huang, Y.L.},
  \bibinfo{author}{Huang, C.Y.}, \bibinfo{author}{Sastry, S.},
  \bibinfo{year}{2011}.
\newblock \bibinfo{title}{Attacks against process control systems: risk
  assessment, detection, and response}, in: \bibinfo{booktitle}{Proceedings of
  the 6th ACM symposium on information, computer and communications security},
  \bibinfo{organization}{ACM}. pp. \bibinfo{pages}{355--366}.
\bibitem[{Corporation(2019)}]{ATTACK}
\bibinfo{author}{Corporation, T.M.}, \bibinfo{year}{2019}.
\newblock \bibinfo{title}{Enterprise matrix}.
\newblock \URLprefix \url{https://attack.mitre.org/matrices/enterprise/}.
\bibitem[{Dufresne(2018)}]{endgame}
\bibinfo{author}{Dufresne, M.}, \bibinfo{year}{2018}.
\newblock \bibinfo{title}{Putting the {MITRE ATT\&CK} evaluation into context}.
\newblock \URLprefix
  \url{https://www.endgame.com/blog/technical-blog/putting-mitre-attck-evaluation-context}.
\bibitem[{Feng et~al.(2015)Feng, Zheng, Hu, Cansever and
  Mohapatra}]{feng2015stealthy}
\bibinfo{author}{Feng, X.}, \bibinfo{author}{Zheng, Z.}, \bibinfo{author}{Hu,
  P.}, \bibinfo{author}{Cansever, D.}, \bibinfo{author}{Mohapatra, P.},
  \bibinfo{year}{2015}.
\newblock \bibinfo{title}{Stealthy attacks meets insider threats: a
  three-player game model}, in: \bibinfo{booktitle}{MILCOM 2015-2015 IEEE
  Military Communications Conference}, \bibinfo{organization}{IEEE}. pp.
  \bibinfo{pages}{25--30}.
\bibitem[{FireEye(2017)}]{FireEye2017}
\bibinfo{author}{FireEye}, \bibinfo{year}{2017}.
\newblock \bibinfo{title}{{Advanced Persistent Threat Groups | FireEye}}.
\newblock \URLprefix
  \url{https://www.fireeye.com/current-threats/apt-groups.html}.
\bibitem[{Friedberg et~al.(2015)Friedberg, Skopik, Settanni and
  Fiedler}]{friedberg2015combating}
\bibinfo{author}{Friedberg, I.}, \bibinfo{author}{Skopik, F.},
  \bibinfo{author}{Settanni, G.}, \bibinfo{author}{Fiedler, R.},
  \bibinfo{year}{2015}.
\newblock \bibinfo{title}{Combating advanced persistent threats: From network
  event correlation to incident detection}.
\newblock \bibinfo{journal}{Computers \& Security} \bibinfo{volume}{48},
  \bibinfo{pages}{35--57}.
\bibitem[{Ghafir et~al.(2018)Ghafir, Hammoudeh, Prenosil, Han, Hegarty, Rabie
  and Aparicio-Navarro}]{ghafir2018detection}
\bibinfo{author}{Ghafir, I.}, \bibinfo{author}{Hammoudeh, M.},
  \bibinfo{author}{Prenosil, V.}, \bibinfo{author}{Han, L.},
  \bibinfo{author}{Hegarty, R.}, \bibinfo{author}{Rabie, K.},
  \bibinfo{author}{Aparicio-Navarro, F.J.}, \bibinfo{year}{2018}.
\newblock \bibinfo{title}{Detection of advanced persistent threat using
  machine-learning correlation analysis}.
\newblock \bibinfo{journal}{Future Generation Computer Systems}
  \bibinfo{volume}{89}, \bibinfo{pages}{349--359}.
\bibitem[{Ghafir et~al.(2019)Ghafir, Kyriakopoulos, Lambotharan,
  Aparicio-Navarro, AsSadhan, BinSalleeh and Diab}]{ghafir2019hidden}
\bibinfo{author}{Ghafir, I.}, \bibinfo{author}{Kyriakopoulos, K.G.},
  \bibinfo{author}{Lambotharan, S.}, \bibinfo{author}{Aparicio-Navarro, F.J.},
  \bibinfo{author}{AsSadhan, B.}, \bibinfo{author}{BinSalleeh, H.},
  \bibinfo{author}{Diab, D.M.}, \bibinfo{year}{2019}.
\newblock \bibinfo{title}{Hidden markov models and alert correlations for the
  prediction of advanced persistent threats}.
\newblock \bibinfo{journal}{IEEE Access} \bibinfo{volume}{7},
  \bibinfo{pages}{99508--99520}.
\bibitem[{Ghafir et~al.(2017)Ghafir, Prenosil, Hammoudeh, Han and
  Raza}]{ghafir2017malicious}
\bibinfo{author}{Ghafir, I.}, \bibinfo{author}{Prenosil, V.},
  \bibinfo{author}{Hammoudeh, M.}, \bibinfo{author}{Han, L.},
  \bibinfo{author}{Raza, U.}, \bibinfo{year}{2017}.
\newblock \bibinfo{title}{Malicious ssl certificate detection: A step towards
  advanced persistent threat defence}, in: \bibinfo{booktitle}{Proceedings of
  the International Conference on Future Networks and Distributed Systems},
  \bibinfo{organization}{ACM}. p.~\bibinfo{pages}{27}.
\bibitem[{Harsanyi(1967)}]{harsanyi1967games}
\bibinfo{author}{Harsanyi, J.C.}, \bibinfo{year}{1967}.
\newblock \bibinfo{title}{Games with incomplete information played by
  ``{B}ayesian'' players, i--iii part i. the basic model}.
\newblock \bibinfo{journal}{Management science} \bibinfo{volume}{14},
  \bibinfo{pages}{159--182}.
\bibitem[{Department~of Homeland~Security(2018)}]{DHS2018}
\bibinfo{author}{Department~of Homeland~Security, D.}, \bibinfo{year}{2018}.
\newblock \bibinfo{title}{NSA/CSS Technical Cyber Threat Framework v2 A REPORT
  FROM: CYBERSECURITY OPERATIONS THE CYBERSECURITY PRODUCTS AND SHARING
  DIVISION}.
\newblock \bibinfo{type}{Technical Report}.
\newblock \URLprefix
  \url{https://www.nsa.gov/Portals/70/documents/what-we-do/cybersecurity/professional-resources/ctr-nsa-css-technical-cyber-threat-framework.pdf}.
\bibitem[{Hor{\'a}k et~al.(2017)Hor{\'a}k, Zhu and
  Bo{\v{s}}ansk{\`y}}]{horak2017manipulating}
\bibinfo{author}{Hor{\'a}k, K.}, \bibinfo{author}{Zhu, Q.},
  \bibinfo{author}{Bo{\v{s}}ansk{\`y}, B.}, \bibinfo{year}{2017}.
\newblock \bibinfo{title}{Manipulating adversary’s belief: A dynamic game
  approach to deception by design for proactive network security}, in:
  \bibinfo{booktitle}{International Conference on Decision and Game Theory for
  Security}, \bibinfo{organization}{Springer}. pp. \bibinfo{pages}{273--294}.
\bibitem[{Huang et~al.(2017)Huang, Chen and Zhu}]{huang2017large}
\bibinfo{author}{Huang, L.}, \bibinfo{author}{Chen, J.}, \bibinfo{author}{Zhu,
  Q.}, \bibinfo{year}{2017}.
\newblock \bibinfo{title}{A large-scale markov game approach to dynamic
  protection of interdependent infrastructure networks}, in:
  \bibinfo{booktitle}{International Conference on Decision and Game Theory for
  Security}, \bibinfo{organization}{Springer}. pp. \bibinfo{pages}{357--376}.
\bibitem[{Huang and Zhu(2018)}]{huang2018analysis}
\bibinfo{author}{Huang, L.}, \bibinfo{author}{Zhu, Q.}, \bibinfo{year}{2018}.
\newblock \bibinfo{title}{Analysis and computation of adaptive defense
  strategies against advanced persistent threats for cyber-physical systems},
  in: \bibinfo{booktitle}{International Conference on Decision and Game Theory
  for Security}, \bibinfo{organization}{Springer}. pp.
  \bibinfo{pages}{205--226}.
\bibitem[{Huang and Zhu(2019a)}]{engagement}
\bibinfo{author}{Huang, L.}, \bibinfo{author}{Zhu, Q.}, \bibinfo{year}{2019}a.
\newblock \bibinfo{title}{Adaptive honeypot engagement through reinforcement
  learning of semi-markov decision processes}, in:
  \bibinfo{booktitle}{International Conference on Decision and Game Theory for
  Security}, \bibinfo{organization}{Springer}. pp. \bibinfo{pages}{196--216}.
\bibitem[{Huang and Zhu(2019b)}]{huang2019adaptive}
\bibinfo{author}{Huang, L.}, \bibinfo{author}{Zhu, Q.}, \bibinfo{year}{2019}b.
\newblock \bibinfo{title}{Adaptive strategic cyber defense for advanced
  persistent threats in critical infrastructure networks}.
\newblock \bibinfo{journal}{ACM SIGMETRICS Performance Evaluation Review}
  \bibinfo{volume}{46}, \bibinfo{pages}{52--56}.
\bibitem[{Hutchins et~al.(2011)Hutchins, Cloppert and
  Amin}]{hutchins2011intelligence}
\bibinfo{author}{Hutchins, E.M.}, \bibinfo{author}{Cloppert, M.J.},
  \bibinfo{author}{Amin, R.M.}, \bibinfo{year}{2011}.
\newblock \bibinfo{title}{Intelligence-driven computer network defense informed
  by analysis of adversary campaigns and intrusion kill chains}.
\newblock \bibinfo{journal}{Leading Issues in Information Warfare \& Security
  Research} \bibinfo{volume}{1}, \bibinfo{pages}{80}.
\bibitem[{Krotofil and C{\'a}rdenas(2013)}]{krotofil2013resilience}
\bibinfo{author}{Krotofil, M.}, \bibinfo{author}{C{\'a}rdenas, A.A.},
  \bibinfo{year}{2013}.
\newblock \bibinfo{title}{Resilience of process control systems to
  cyber-physical attacks}, in: \bibinfo{booktitle}{Nordic Conference on Secure
  IT Systems}, \bibinfo{organization}{Springer}. pp. \bibinfo{pages}{166--182}.
\bibitem[{La et~al.(2016)La, Quek, Lee, Jin and Zhu}]{la2016deceptive}
\bibinfo{author}{La, Q.D.}, \bibinfo{author}{Quek, T.Q.}, \bibinfo{author}{Lee,
  J.}, \bibinfo{author}{Jin, S.}, \bibinfo{author}{Zhu, H.},
  \bibinfo{year}{2016}.
\newblock \bibinfo{title}{Deceptive attack and defense game in honeypot-enabled
  networks for the internet of things}.
\newblock \bibinfo{journal}{IEEE Internet of Things Journal}
  \bibinfo{volume}{3}, \bibinfo{pages}{1025--1035}.
\bibitem[{Li et~al.(2018)Li, Yang, Xiong, Wen and Tang}]{li2018defending}
\bibinfo{author}{Li, P.}, \bibinfo{author}{Yang, X.}, \bibinfo{author}{Xiong,
  Q.}, \bibinfo{author}{Wen, J.}, \bibinfo{author}{Tang, Y.Y.},
  \bibinfo{year}{2018}.
\newblock \bibinfo{title}{Defending against the advanced persistent threat: An
  optimal control approach}.
\newblock \bibinfo{journal}{Security and Communication Networks}
  \bibinfo{volume}{2018}.
\bibitem[{LLC(2018)}]{Jeff2000}
\bibinfo{author}{LLC, P.I.}, \bibinfo{year}{2018}.
\newblock \bibinfo{title}{2018 cost of data breach study}.
\bibitem[{L{\"{o}}fberg(2004)}]{Lofberg2004}
\bibinfo{author}{L{\"{o}}fberg, J.}, \bibinfo{year}{2004}.
\newblock \bibinfo{title}{Yalmip : A toolbox for modeling and optimization in
  {MATLAB}}, in: \bibinfo{booktitle}{In Proceedings of the CACSD Conference},
  \bibinfo{address}{Taipei, Taiwan}.
\bibitem[{Marchetti et~al.(2016)Marchetti, Pierazzi, Colajanni and
  Guido}]{marchetti2016analysis}
\bibinfo{author}{Marchetti, M.}, \bibinfo{author}{Pierazzi, F.},
  \bibinfo{author}{Colajanni, M.}, \bibinfo{author}{Guido, A.},
  \bibinfo{year}{2016}.
\newblock \bibinfo{title}{Analysis of high volumes of network traffic for
  advanced persistent threat detection}.
\newblock \bibinfo{journal}{Computer Networks} \bibinfo{volume}{109},
  \bibinfo{pages}{127--141}.
\bibitem[{Messaoud et~al.(2016)Messaoud, Guennoun, Wahbi and
  Sadik}]{messaoud2016advanced}
\bibinfo{author}{Messaoud, B.I.}, \bibinfo{author}{Guennoun, K.},
  \bibinfo{author}{Wahbi, M.}, \bibinfo{author}{Sadik, M.},
  \bibinfo{year}{2016}.
\newblock \bibinfo{title}{Advanced persistent threat: New analysis driven by
  life cycle phases and their challenges}, in: \bibinfo{booktitle}{2016
  International Conference on Advanced Communication Systems and Information
  Security (ACOSIS)}, \bibinfo{organization}{IEEE}. pp. \bibinfo{pages}{1--6}.
\bibitem[{Milajerdi and Kharrazi(2015)}]{milajerdi2015composite}
\bibinfo{author}{Milajerdi, S.M.}, \bibinfo{author}{Kharrazi, M.},
  \bibinfo{year}{2015}.
\newblock \bibinfo{title}{A composite-metric based path selection technique for
  the tor anonymity network}.
\newblock \bibinfo{journal}{Journal of Systems and Software}
  \bibinfo{volume}{103}, \bibinfo{pages}{53--61}.
\bibitem[{Mitnick and Simon(2011)}]{mitnick2011art}
\bibinfo{author}{Mitnick, K.D.}, \bibinfo{author}{Simon, W.L.},
  \bibinfo{year}{2011}.
\newblock \bibinfo{title}{The art of deception: Controlling the human element
  of security}.
\newblock \bibinfo{publisher}{John Wiley \& Sons}.
\bibitem[{Molok et~al.(2010)Molok, Chang and Ahmad}]{molok2010information}
\bibinfo{author}{Molok, N.N.A.}, \bibinfo{author}{Chang, S.},
  \bibinfo{author}{Ahmad, A.}, \bibinfo{year}{2010}.
\newblock \bibinfo{title}{Information leakage through online social networking:
  Opening the doorway for advanced persistence threats} .
\bibitem[{Morris and Gao(2013)}]{morris2013industrial}
\bibinfo{author}{Morris, T.H.}, \bibinfo{author}{Gao, W.},
  \bibinfo{year}{2013}.
\newblock \bibinfo{title}{Industrial control system cyber attacks}, in:
  \bibinfo{booktitle}{Proceedings of the 1st International Symposium on ICS \&
  SCADA Cyber Security Research}, pp. \bibinfo{pages}{22--29}.
\bibitem[{Nguyen et~al.()Nguyen, Wang, Sinha and Wellman}]{NguyenDeceptionIF}
\bibinfo{author}{Nguyen, T.H.}, \bibinfo{author}{Wang, Y.},
  \bibinfo{author}{Sinha, A.}, \bibinfo{author}{Wellman, M.P.}, .
\newblock \bibinfo{title}{Deception in finitely repeated security games}.
\bibitem[{Nissim et~al.(2015)Nissim, Cohen, Glezer and
  Elovici}]{nissim2015detection}
\bibinfo{author}{Nissim, N.}, \bibinfo{author}{Cohen, A.},
  \bibinfo{author}{Glezer, C.}, \bibinfo{author}{Elovici, Y.},
  \bibinfo{year}{2015}.
\newblock \bibinfo{title}{Detection of malicious pdf files and directions for
  enhancements: A state-of-the art survey}.
\newblock \bibinfo{journal}{Computers \& Security} \bibinfo{volume}{48},
  \bibinfo{pages}{246--266}.
\bibitem[{Pawlick et~al.(2018)Pawlick, Chen and Zhu}]{pawlick2018istrict}
\bibinfo{author}{Pawlick, J.}, \bibinfo{author}{Chen, J.},
  \bibinfo{author}{Zhu, Q.}, \bibinfo{year}{2018}.
\newblock \bibinfo{title}{istrict: An interdependent strategic trust mechanism
  for the cloud-enabled internet of controlled things}.
\newblock \bibinfo{journal}{arXiv preprint arXiv:1805.00403} .
\bibitem[{Pawlick et~al.(2017)Pawlick, Colbert and Zhu}]{pawlick2017game}
\bibinfo{author}{Pawlick, J.}, \bibinfo{author}{Colbert, E.},
  \bibinfo{author}{Zhu, Q.}, \bibinfo{year}{2017}.
\newblock \bibinfo{title}{A game-theoretic taxonomy and survey of defensive
  deception for cybersecurity and privacy}.
\newblock \bibinfo{journal}{arXiv preprint arXiv:1712.05441} .
\bibitem[{Ricker(1996)}]{ricker1996decentralized}
\bibinfo{author}{Ricker, N.L.}, \bibinfo{year}{1996}.
\newblock \bibinfo{title}{Decentralized control of the tennessee eastman
  challenge process}.
\newblock \bibinfo{journal}{Journal of Process Control} \bibinfo{volume}{6},
  \bibinfo{pages}{205--221}.
\bibitem[{Rowe et~al.(2007)Rowe, Custy and Duong}]{rowe2007defending}
\bibinfo{author}{Rowe, N.C.}, \bibinfo{author}{Custy, E.J.},
  \bibinfo{author}{Duong, B.T.}, \bibinfo{year}{2007}.
\newblock \bibinfo{title}{Defending cyberspace with fake honeypots} .
\bibitem[{Sahoo et~al.(2017)Sahoo, Liu and Hoi}]{sahoo2017malicious}
\bibinfo{author}{Sahoo, D.}, \bibinfo{author}{Liu, C.}, \bibinfo{author}{Hoi,
  S.C.}, \bibinfo{year}{2017}.
\newblock \bibinfo{title}{Malicious url detection using machine learning: a
  survey}.
\newblock \bibinfo{journal}{arXiv preprint arXiv:1701.07179} .
\bibitem[{Shoham and Leyton-Brown(2008)}]{shoham2008multiagent}
\bibinfo{author}{Shoham, Y.}, \bibinfo{author}{Leyton-Brown, K.},
  \bibinfo{year}{2008}.
\newblock \bibinfo{title}{Multiagent systems: Algorithmic, game-theoretic, and
  logical foundations}.
\newblock \bibinfo{publisher}{Cambridge University Press}.
\bibitem[{Sigholm and Bang(2013)}]{sigholm2013towards}
\bibinfo{author}{Sigholm, J.}, \bibinfo{author}{Bang, M.},
  \bibinfo{year}{2013}.
\newblock \bibinfo{title}{Towards offensive cyber counterintelligence: Adopting
  a target-centric view on advanced persistent threats}, in:
  \bibinfo{booktitle}{2013 European Intelligence and Security Informatics
  Conference}, \bibinfo{organization}{IEEE}. pp. \bibinfo{pages}{166--171}.
\bibitem[{Tawarmalani and Sahinidis(2005)}]{ts}
\bibinfo{author}{Tawarmalani, M.}, \bibinfo{author}{Sahinidis, N.V.},
  \bibinfo{year}{2005}.
\newblock \bibinfo{title}{{A polyhedral branch-and-cut approach to global
  optimization}}.
\newblock \bibinfo{journal}{Mathematical Programming} \bibinfo{volume}{103},
  \bibinfo{pages}{225--249}.
\bibitem[{Team(2017)}]{processinjection}
\bibinfo{author}{Team, M.D.A.R.}, \bibinfo{year}{2017}.
\newblock \bibinfo{title}{Detecting stealthier cross-process injection
  techniques with windows defender atp: Process hollowing and atom bombing}.
\newblock \URLprefix \url{https://bit.ly/2nVWDQd}.
\bibitem[{Van~Dijk et~al.(2013)Van~Dijk, Juels, Oprea and
  Rivest}]{van2013flipit}
\bibinfo{author}{Van~Dijk, M.}, \bibinfo{author}{Juels, A.},
  \bibinfo{author}{Oprea, A.}, \bibinfo{author}{Rivest, R.L.},
  \bibinfo{year}{2013}.
\newblock \bibinfo{title}{Flipit: The game of {``stealthy takeover''}}.
\newblock \bibinfo{journal}{Journal of Cryptology} \bibinfo{volume}{26},
  \bibinfo{pages}{655--713}.
\bibitem[{Yang et~al.(2018)Yang, Li, Zhang, Yang, Xiang and
  Zhou}]{yang2018effective}
\bibinfo{author}{Yang, L.X.}, \bibinfo{author}{Li, P.}, \bibinfo{author}{Zhang,
  Y.}, \bibinfo{author}{Yang, X.}, \bibinfo{author}{Xiang, Y.},
  \bibinfo{author}{Zhou, W.}, \bibinfo{year}{2018}.
\newblock \bibinfo{title}{Effective repair strategy against advanced persistent
  threat: A differential game approach}.
\newblock \bibinfo{journal}{IEEE Transactions on Information Forensics and
  Security} \bibinfo{volume}{14}, \bibinfo{pages}{1713--1728}.
\bibitem[{Zhang et~al.(2015)Zhang, Zheng and Shroff}]{zhang2015game}
\bibinfo{author}{Zhang, M.}, \bibinfo{author}{Zheng, Z.},
  \bibinfo{author}{Shroff, N.B.}, \bibinfo{year}{2015}.
\newblock \bibinfo{title}{A game theoretic model for defending against stealthy
  attacks with limited resources}, in: \bibinfo{booktitle}{International
  Conference on Decision and Game Theory for Security},
  \bibinfo{organization}{Springer}. pp. \bibinfo{pages}{93--112}.
\bibitem[{Zhu and Rass(2018)}]{zhu2018multi}
\bibinfo{author}{Zhu, Q.}, \bibinfo{author}{Rass, S.}, \bibinfo{year}{2018}.
\newblock \bibinfo{title}{On multi-phase and multi-stage game-theoretic
  modeling of advanced persistent threats}.
\newblock \bibinfo{journal}{IEEE Access} \bibinfo{volume}{6},
  \bibinfo{pages}{13958--13971}.

\end{thebibliography}


%

\end{document}